\documentclass[a4paper]{article}

\usepackage{a4wide}
\usepackage[UKenglish]{babel}
\usepackage{authblk}
\usepackage{graphicx,subfigure,caption}
\graphicspath{{./}{figures/}}
\usepackage[colorlinks=true,linkcolor=blue,citecolor=red]{hyperref}
\usepackage{xcolor}
\usepackage{amssymb,amsthm,empheq,bbold}
\usepackage[inline]{enumitem}
\usepackage[title]{appendix}

\DeclarePairedDelimiter\ave{\langle}{\rangle}
\newcommand{\cM}{\mathcal{M}}
\newcommand{\N}{\mathbb{N}}
\newcommand{\Prob}[1]{\operatorname{Prob}(#1)}
\newcommand{\R}{\mathbb{R}}
\newcommand{\cT}{\mathcal{T}}
\newcommand{\Z}{\mathbb{Z}}

\newtheorem{assumption}{Assumption}[section]
\newtheorem{proposition}[assumption]{Proposition}
\newtheorem{theorem}[assumption]{Theorem}
\newtheorem{lemma}[assumption]{Lemma}
\newtheorem{corollary}[assumption]{Corollary}
\theoremstyle{remark}\newtheorem{remark}[assumption]{Remark}
\allowdisplaybreaks
		
\title{Probabilistic modeling of car traffic accidents}
\author[1]{Simone G\"{o}ttlich}
\author[1]{Thomas Schillinger}
\author[2]{Andrea Tosin}
\affil[1]{{\footnotesize University of Mannheim, School of Business Informatics and Mathematics, 68159 Mannheim, Germany}}
\affil[2]{{\footnotesize Politecnico di Torino, Department of Mathematical Sciences ``G. L. Lagrange'', Torino, Italy}}
\date{}
	
\begin{document}
\maketitle
		
\begin{abstract}
We introduce a counting process to model the random occurrence in time of car traffic accidents, taking into account some aspects of the self-excitation typical of this phenomenon. By combining methods from probability and differential equations, we study this stochastic process in terms of its statistical moments and large-time trend. Moreover, we derive analytically the probability density functions of the times of occurrence of traffic accidents and of the time elapsing between two consecutive accidents. Finally, we demonstrate the suitability of our modelling approach by means of numerical simulations, which address also a comparison with real data of weekly trends of traffic accidents.

\medskip
			
			
\medskip
			
\end{abstract}
		
\section{Introduction}
\label{sect:intro}
Understanding the occurrence of accidents is important in terms of increased travel time due to congestion, material damage to vehicles, and -- most importantly -- the safety of travellers. Although the number of road fatalities has decreased in most European countries over the last fifty years, road accidents are still a serious problem and have become an increasingly important issue in today's controversial topics. In particular, traffic accidents lead to congestion and are therefore a cause of stress and increased travel times, as well as of higher vehicle emissions.

There is a large variety of mathematical perspectives on traffic accidents. Considering traffic flow models based on first order hyperbolic conservation laws, traffic accidents were introduced in \cite{Goettlich.2020,Goettlich.2021,Goettlich.2024}. These works consider coupled models where traffic flow and traffic accidents influence each other in a bidirectional way. This idea was transferred to second order traffic dynamics in \cite{Chiarello.2023}. Using machine learning techniques, investigations on the prediction of traffic accidents \cite{garcia2018,Zhao2019}, the number of accidents on particular road segments \cite{Gataric.2023}, and also further accident-related parameters such as their severity \cite{Alhaek.2024} have been carried out. Further approaches involve Bayesian networks \cite{Mora2017,Zou2017} and kinetic models \cite{freguglia2017}.

\begin{figure}[!t]
    \centering
    \includegraphics[width=\linewidth]{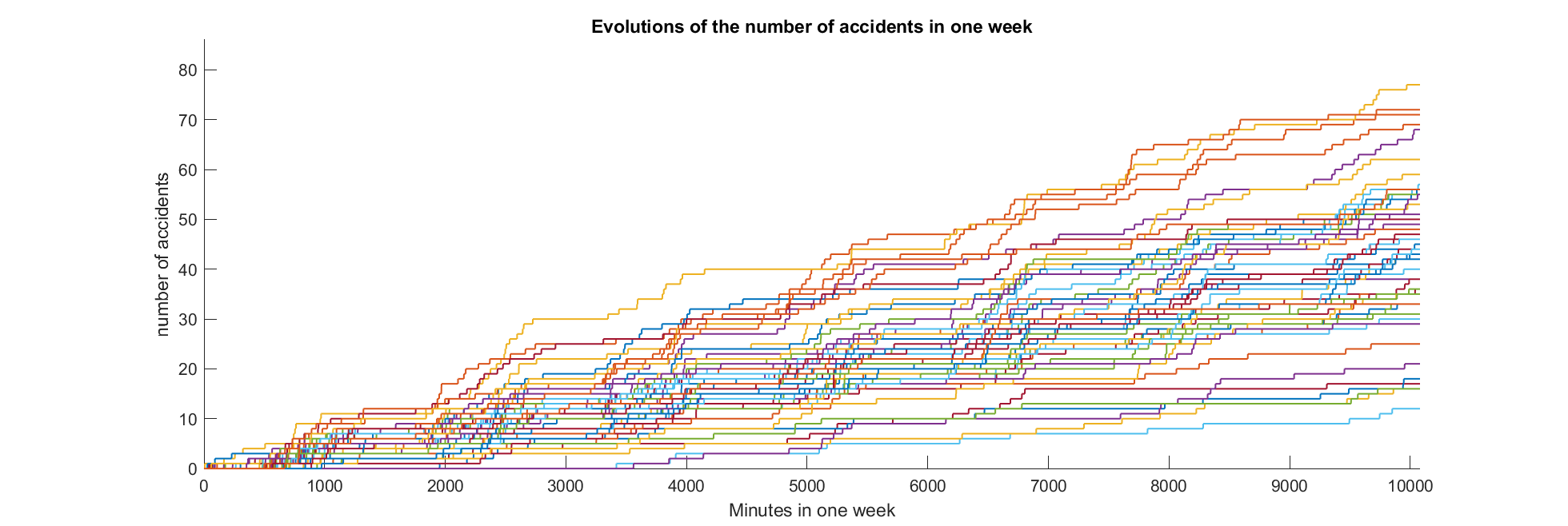}
    \caption{Weekly evolution in time of the number of accidents in an area in the UK for all weeks of 2020. For further details on the data, see Section~\ref{sec:realData}.}
    \label{fig:motivation}
\end{figure}

Conversely, a \textit{probabilistic modelling} of car traffic accidents is still largely missing. In order to motivate it, we observe in Figure~\ref{fig:motivation} some typical time trajectories of the number of accidents based on British accident data\footnote{\url{https://data.gov.uk/dataset/cb7ae6f0-4be6-4935-9277-47e5ce24a11f/road-safety-data} -- Accessed 11th September 2024.}. Each curve represents one week in 2020 and, as it is apparent from the plot, features a different growing trend with several jumps typical of a counting process. Instead of investigating the detailed causes of such trends, a probabilistic modelling approach should aim to explain \textit{aggregately} the origin of time trajectories qualitatively similar to those of Figure~\ref{fig:motivation} out of more \textit{elementary universal principles}. One of them is, for instance, the self-adaptation of the process based on previous accidents, in particular the fact that new accidents get more probable after an accident has occurred -- the so-called \textit{self-excitation} of the process.

To describe mathematically the jump intensity of counting processes like the one reported in Figure~\ref{fig:motivation}, different approaches exist in the literature. When considering the occurrence of traffic accidents from a statistical point of view, Poisson models are a frequent choice~\cite{Djauhari.2002,Gustavsson.1976}. Nevertheless, as data comparisons depict, their applicability is limited due to larger variances in real world data than those models predict (over-dispersion). In~\cite{Lord.2010,Sipos.2017} this issue is further elaborated and first approaches to bypass it are presented.

Building on~\cite{Goettlich.2024}, where the influence of previous accidents on the probability of further accidents was investigated using a Hawkes process, in this paper we follow instead the idea that, as mentioned before, accidents are driven by some self-excitation property. The latter, first considered for earthquakes \cite{Hawkes.1971}, has been investigated recently in the context of traffic accidents in~\cite{alaimo2024JABES,errais2022AOR,Kalair.2020,Motagi.2023}. We refer to~\cite{Li.2018} for a simulation-driven investigation of statistical properties and a maximum-likelihood analysis of accident parameters. However, since the Hawkes process has a jump intensity that typically decays rapidly in time, it is only useful for the description of short-term accident patterns and not for the explanation of the main effects depicted in Figure \ref{fig:motivation}.

To overcome this difficulty, while keeping an eye on the amenability of our model to insightful analytical investigations, we propose a stochastic counting process featuring a generalised version of the self-excitation property recalled above. In particular, its jump probability consists of a time-dependent background intensity supplemented by a term proportional to the cumulative number of accidents occurred up to the current time, the so-called self-excitation intensity. For both terms we consider prototypical assumptions and derive the corresponding characteristics of the stochastic process. Specifically, by combining methods from probability and differential equations in a spirit reminiscent of that of the kinetic theory, we estimate the large-time trend of the process, the behaviour of its statistical moments, and, under certain assumptions, we even provide an explicit expression of its law. Particularly relevant are moreover the time instants of occurrence of traffic accidents as well as the time elapsing between two consecutive accidents (interaccident time). For both quantities, which are random variables linked to the main counting process, we provide explicit probabilistic characterisations, showing that in general their distributions do not follow those typically postulated in the more standard probabilistic descriptions mentioned above. This is in line with the observed mismatch of e.g., Poisson models when employed to describe mathematically traffic accidents.

In more detail, the paper is structured as follows. Section~\ref{sect:model.accident_distr} introduces our stochastic process modelling the occurrence of car accidents and investigates the basic properties of its probability distribution. Making use of such results, Section~\ref{sect:acc_times} addresses the study of the distribution of accident and interaccident times. Next, for different choices of the background and self-excitation intensity functions complying with the theory previously developed, Section~\ref{sec:realData} addresses a real data study aimed at identifying proper choices of the intensity functions which make our model reproduce qualitatively the benchmark trends displayed in Figure~\ref{fig:motivation}. Finally, Section~\ref{sect:conclusions} summarises the main contributions of the paper.

\section{Modelling the accident distribution}
\label{sect:model.accident_distr}
Let us consider a stream of vehicles along a generic stretch of road. We denote by $\{N_t,\,t\geq 0\}$ the stochastic process counting the number of accidents over time; in particular, $N_t\in\N$ is the random variable yielding the cumulative number of accidents occurred up to time $t>0$, being $N_0=0$ the (deterministic) initial condition.

Upon introducing a small time step $0<\Delta{t}\ll 1$, we consider the following discrete-in-time evolution rule of $N_t$:
\begin{equation}
    N_{t+\Delta{t}}=N_t+H^{\Delta{t}}_t,
    \label{eq:Nt.evolution}
\end{equation}
where $H^{\Delta{t}}_t\in\N$ is a random variable counting the number of accidents which occur in the time interval $(t,\,t+\Delta{t}]$. We assume the following probabilistic model:
\begin{equation}
    \Prob{H^{\Delta{t}}_t=h}=
    \begin{cases}
        1-\lambda^\ast(t)\Delta{t}+o(\Delta{t}) & \text{if } h=0 \\
        \lambda^\ast(t)\Delta{t}+o(\Delta{t}) & \text{if } h=1 \\
        o(\Delta{t}) & \text{if } h\geq 2,
    \end{cases}
    \qquad h\in\N,
    \label{eq:Prob.H}
\end{equation}
where the function $\lambda^\ast:[0,\,+\infty)\to\R_+$, to be specified, is the \textit{intensity} of the process $\{N_t,\,t\geq 0\}$.

\begin{figure}[!t]
    \centering
    \includegraphics[width=.7\linewidth]{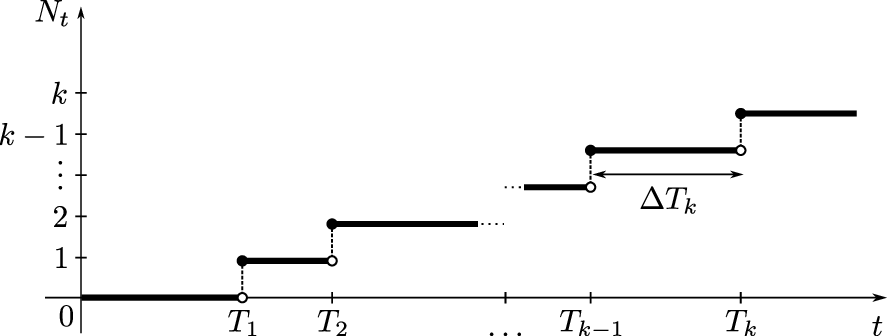}
    \caption{A prototypical trajectory of the process $\{N_t,\,t\geq 0\}$.}
    \label{fig:Nt}
\end{figure}

The process $\{N_t,\,t\geq 0\}$ just defined belongs to the class of the \textit{counting processes}. Figure~\ref{fig:Nt} displays a prototypical trajectory of it, which is piecewise constant in the set of positive integers. The process jumps from one positive integer to the next one at some random times $\{T_k\}_{k\geq 1}$, $0<T_1<T_2<\ldots<T_{k-1}<T_k<\ldots$, called \textit{jump} (or \textit{arrival}) \textit{times}. The time between two consecutive jump times, during which the process remains in a specific state (integer value), is called a \textit{holding} (or \textit{interarrival}) \textit{time} and is defined as $\Delta{T}_k:=T_k-T_{k-1}$. Typical examples of counting processes are the \textit{renewal processes}, which feature independent and identically distributed holding times. For example, if in~\eqref{eq:Prob.H} one considers a constant intensity $\lambda^\ast>0$ one obtains a renewal process $\{N_t,\,t\geq 0\}$ whose holding times are independent and exponentially distributed with parameter $\lambda^\ast$, i.e. a \textit{Poisson process} (cf. Section~\ref{sect:aggregateInterTimes.constant_lambda_mu} in the Supplementary Material).

In our case, however, we need to consider more elaborated forms of the intensity $\lambda^\ast$ to include in the model the \textit{self-excitation} property of car accidents mentioned in Section~\ref{sect:intro}. This typically destroys the independence and identical distribution of the holding times, so that the resulting counting process $\{N_t,\,t\geq 0\}$ is no longer a renewal process and needs therefore to be analysed in a dedicated manner.

In applications ranging e.g., from seismology to epidemiology and mathematical finance the prototype of self-exciting stochastic processes is the \textit{Hawkes process}~\cite{Hawkes.1971}, in which the intensity $\lambda^\ast(t)$ depends on the history of the process $\{N_s,\,s\geq 0\}$ up to time $t$ as follows:
\begin{equation}
    \lambda^\ast(t)=\lambda(t)+\int_0^t\mu(t-s)\,dN_s,
    \label{eq:lambda*.Hawkes}
\end{equation}
where $\lambda,\,\mu:[0,\,+\infty)\to\R_+$ are the so-called \textit{background} and \textit{self-excitation intensities}, respectively. The integral appearing in~\eqref{eq:lambda*.Hawkes} is a stochastic integral with respect to the process $\{N_t,\,t\geq 0\}$. Using the fact that the latter is piecewise constant with jumps of unitary amplitude, the intensity $\lambda^\ast$ may be written equivalently as
\begin{equation}
    \lambda^\ast(t)=\lambda(t)+\sum_{k\,:\,T_k\leq t}\mu(t-T_k)
    \label{eq:lambda_ast.acc_times}
\end{equation}
at least for a continuous self-excitation intensity function $\mu$. In particular, if $\mu$ is constant we obtain the special form
\begin{equation}
    \lambda^\ast(t)=\lambda(t)+\mu N_t.
    \label{eq:Hawkes.mu_const}
\end{equation}

In order to bypass the complications of a non-local (in time) probabilistic model like~\eqref{eq:lambda*.Hawkes}, taking inspiration from~\eqref{eq:Hawkes.mu_const} here we define instead
\begin{equation}
    \lambda^\ast(t)=\lambda(t)+\mu(t)N_t,
    \label{eq:our_lambda*}
\end{equation}
so as to allow for a variable self-excitation intensity but with a local-in-time effect.

\begin{remark} \label{rem:generalised_Hawkes}
Some authors, cf. e.g.,~\cite{roueff2016SPA}, generalise the definition~\eqref{eq:lambda*.Hawkes} of the intensity $\lambda^\ast$ of a Hawkes process by admitting that the self-excitation intensity $\mu$ may be \textit{time-varying}, i.e.,
$$ \lambda^\ast(t)=\lambda(t)+\int_0^t\mu(t-s;\,t)\,dN_s. $$
This assumption corresponds to saying that the self-excitation does not only depend on the time elapsed since each event but also on time-varying contextual characteristics of the phenomenon under consideration. From an alternative point of view, the self-excitation intensity is \textit{parameterised} by $t$, which may be understood as a way to model the influence over time of the background on the \textit{type} of self-excitation. Within this framework, we observe that if $\mu$ is constant with respect to the elapsed time $t-s$ one recovers exactly our model~\eqref{eq:our_lambda*}. Therefore, the latter may be possibly regarded as a particular instance of a (generalised) Hawkes process, when the self-excitation is mainly dictated by the background.
\end{remark}

On the whole, \eqref{eq:Nt.evolution},~\eqref{eq:Prob.H},~\eqref{eq:our_lambda*} fully characterise our discrete-in-time stochastic accident model.

\begin{remark} \label{rem:physics.lambda*}
From the physical point of view,~\eqref{eq:our_lambda*} models a system which accumulates ``energy'', viz. \textit{accident risk}, at each accident event without ever releasing it, as it will become apparent in Remark~\ref{rem:energy}. In reality, such an energy should be released eventually, a feature which is not caught by~\eqref{eq:our_lambda*}. Therefore, one should expect that such an ``accumulation-without-release'' property may be appropriate to describe a sequence of traffic accidents on a time scale compatible with a \textit{limited} physical time horizon. As we will ascertain with the computational analysis of Section~\ref{sec:realData}, this is indeed true for medium-size time horizons of the order of weeks.

We stress that, in general, it is reasonable to study traffic accidents over limited time horizons. Indeed, the physics of traffic is much less established than that of other physical phenomena, such as those mentioned before, whose study may instead be extended to larger time horizons because it may be grounded on a sounder understanding of the underlying physical principles. From this perspective, a medium-size time horizon can foster the meaningfulness of the assumption that $\mu$ is determined essentially by the background, being instead constant with respect to the time elapsed since each accident (cf. Remark~\ref{rem:generalised_Hawkes}).
\end{remark}

Now, let
\begin{align}
f(n,t):=\Prob{N_t=n}, \qquad n\in\N.
\label{eq:distributionFunction}
\end{align}
From the stochastic model~\eqref{eq:Nt.evolution},~\eqref{eq:Prob.H},~\eqref{eq:our_lambda*} we may obtain an evolution equation for $f$ in the continuous time limit $\Delta{t}\to 0^+$ by invoking classical arguments of the kinetic theory, cf. e.g.,~\cite{pareschi2013BOOK}. To this purpose, let $\varphi:\N\to\R$ be an arbitrary test function and let us denote by $\ave{\cdot}$ the expectation operator. Taking the expectation of the quantity $\varphi(N_{t+\Delta{t}})$ we get, from~\eqref{eq:Nt.evolution},~\eqref{eq:Prob.H},
$$ \ave{\varphi(N_{t+\Delta{t}})}=\ave{\varphi(N_t+H^{\Delta{t}}_t)}
    =\ave{\varphi(N_t)}+\ave*{\lambda^\ast(t)\bigl(\varphi(N_t+1)-\varphi(N_t)\bigr)}\Delta{t}+o(\Delta{t}), $$
whence
$$ \frac{\ave{\varphi(N_{t+\Delta{t}})}-\ave{\varphi(N_t)}}{\Delta{t}}=
    \ave*{\lambda^\ast(t)\bigl(\varphi(N_t+1)-\varphi(N_t)\bigr)}+o(1) $$
and formally, in the limit $\Delta{t}\to 0^+$,
$$ \frac{d}{dt}\ave{\varphi(N_t)}=\ave{\lambda^\ast(t)\bigl(\varphi(N_t+1)-\varphi(N_t)\bigr)}. $$
Using $f$ to compute the remaining expectations we find further
\begin{equation}
    \frac{d}{dt}\sum_{i=0}^\infty\varphi(i)f(i,t)=\sum_{i=0}^\infty\left(\lambda(t)+\mu(t)i\right)\bigl(\varphi(i+1)-\varphi(i)\bigr)f(i,t),
    \label{eq:Boltz_f.weak}
\end{equation}
where we have used~\eqref{eq:our_lambda*} to express $\lambda^\ast$ in terms of $\lambda$, $\mu$. This equation is required to hold for every test function $\varphi$, therefore it can be considered the \textit{weak form} of the evolution equation for the distribution function $f$.

Recovering the corresponding \textit{strong form} is a matter of choosing $\varphi$ conveniently. In particular, let us fix $n\in\N$ and let us consider a test function $\varphi$ such that $\varphi(n)=1$ while $\varphi(i)=0$ whenever $i\neq n$. Then from~\eqref{eq:Boltz_f.weak} we get
\begin{equation}
    \partial_tf(n,t)=\bigl(\lambda(t)+\mu(t)(n-1)\bigr)f(n-1,t)-\bigl(\lambda(t)+\mu(t)n\bigr)f(n,t)
    \label{eq:Boltz_f.strong}
\end{equation}
for $n=0,\,1,\,2,\,\dots$, which provides a set of evolution equations for the probabilities $\{f(n,t),\,n\in\N\}$, i.e. the law of $N_t$ for each $t\geq 0$.

Assuming that at the initial time no accidents have occurred yet, the natural initial condition for~\eqref{eq:Boltz_f.strong} at $t=0$ is
\begin{equation}
    f(n,0)=
        \begin{cases}
            1 & \text{if } n=0 \\
            0 & \text{if } n>0,
        \end{cases}
    \label{eq:f.init_cond}
\end{equation}
which corresponds to $N_0$ having a law given by $\delta_0(n)$.

\begin{remark} \label{rem:f.n_neg}
For subsequent developments, it may be convenient to understand $f(n,t)$ as defined formally also for negative integers $n$, hence on the whole $\Z$, with $f(n,t)=0$ for all $n<0$ and all $t\geq 0$. Notice that~\eqref{eq:Boltz_f.strong} evaluated at $n<0$ is consistent with this interpretation, indeed it returns $\partial_tf(n,t)=0$ for all $n<0$.
\end{remark}

\subsection{Basic properties of the distribution function~\texorpdfstring{$\boldsymbol{f}$}{}} \label{sect:basic_prop.f}
We begin by investigating some elementary properties of the distribution function $f$, which can be deduced directly from the initial-value problem~\eqref{eq:Boltz_f.strong}-\eqref{eq:f.init_cond}.

Let us define
\begin{equation}
    \Lambda(t):=\int_0^t\lambda(s)\,ds, \qquad M(t):=\int_0^t\mu(s)\,ds.
    \label{eq:Lambda_M}
\end{equation}

\begin{assumption} \label{ass:lambda_mu_bounded}
We assume that $\lambda$, $\mu$ are bounded in every interval $[0,\,t]$, $t>0$. Then so are $\Lambda$, $M$.
\end{assumption}

Notice that, owing to Assumption~\ref{ass:lambda_mu_bounded}, the mappings $t\mapsto\Lambda(t)$, $t\mapsto M(t)$ are continuous. Moreover, since $\lambda$, $\mu$ are non-negative, they are also non-decreasing.

From~\eqref{eq:Boltz_f.strong}, for $n=0$ and taking Remark~\ref{rem:f.n_neg} into account we get $\partial_tf(0,t)=-\lambda(t)f(0,t)$, whence, owing to~\eqref{eq:f.init_cond},
\begin{equation}
    f(0,t)=e^{-\Lambda(t)}.
    \label{eq:f(0,t)}
\end{equation}
For $n>0$, rewriting~\eqref{eq:Boltz_f.strong} as
$$ \partial_t\!\left(e^{\Lambda(t)+M(t)n}f(n,t)\right)=e^{\Lambda(t)+M(t)n}\bigl(\lambda(t)+\mu(t)(n-1)\bigr)f(n-1,t) $$
and integrating over $[0,\,t]$, $t>0$, we obtain instead the recursive relationship
\begin{equation}
    f(n,t)=e^{-(\Lambda(t)+M(t)n)}\int_0^te^{\Lambda(s)+M(s)n}\bigl(\lambda(s)+\mu(s)(n-1)\bigr)f(n-1,s)\,ds,
    \label{eq:f.recurrence}
\end{equation}
$n=1,\,2,\dots$, whence we may prove that~\eqref{eq:Boltz_f.strong} is consistent with its solutions being probability distributions:
\begin{proposition} \label{prop:f.prob_distr}
Assuming~\eqref{eq:f.init_cond}, it results
$$ 0\leq f(n,t)\leq 1, \qquad \sum_{i=0}^\infty f(i,t)=1 $$
for all $n\in\N$ and all $t>0$.
\end{proposition}
\begin{proof}
First, we show that $0\leq f(n,t)\leq 1$ for all $n\in\N$ and all $t>0$. For $n=0$, this is evident from~\eqref{eq:f(0,t)}. Assume now that for a certain $n>0$ we have $0\leq f(n-1,t)\leq 1$ for all $t>0$. Then from~\eqref{eq:f.recurrence} we see straightforwardly that $f(n,t)\geq 0$ for all $t>0$ and moreover:
\begin{align*}
    f(n,t) &\leq e^{-(\Lambda(t)+M(t)n)}\int_0^te^{\Lambda(s)+M(s)n}\bigl(\lambda(s)+\mu(s)(n-1)\bigr)\,ds \\
    &= e^{-(\Lambda(t)+M(t)n)}\left[\left(e^{\Lambda(s)+M(s)n}\right\vert_0^t-\int_0^te^{\Lambda(s)+M(s)n}\mu(s)\,ds\right] \\
\intertext{whence, considering that $\Lambda(0)=M(0)=0$ by definition, cf.~\eqref{eq:Lambda_M},}
    &= 1-e^{-(\Lambda(t)+M(t)n)}\left[1+\int_0^te^{\Lambda(s)+M(s)n}\mu(s)\,ds\right]\leq 1.
\end{align*}
The thesis follows then by induction on $n$.

Second, we show that $\sum_{i=0}^\infty f(i,t)=1$ for all $t>0$. From~\eqref{eq:Boltz_f.weak} with $\varphi\equiv 1$ we discover $\frac{d}{dt}\sum_{i=0}^\infty f(i,t)=0$, whence the thesis follows owing to~\eqref{eq:f.init_cond}.
\end{proof}

We may take advantage of the weak form~\eqref{eq:Boltz_f.weak} of the evolution equation for $f$ to study some representative statistical quantities of the process $\{N_t,\,t\geq 0\}$.

Let
\begin{equation}
    m(t):=\sum_{n=0}^\infty nf(n,t)
    \label{eq:m}
\end{equation}
be the \textit{mean number of accidents} at time $t$. Choosing $\varphi(i)=i$ in~\eqref{eq:Boltz_f.weak} we obtain that $m$ satisfies the equation
\begin{equation}
    \frac{dm}{dt}=\lambda(t)+\mu(t)m
    \label{eq:m.ODE}
\end{equation}
complemented with the initial condition $m(0)=0$, cf.~\eqref{eq:f.init_cond}, whence we determine explicitly
\begin{align}
    m(t)=e^{M(t)}\int_0^te^{-M(s)}\lambda(s)\,ds.
    \label{eq:m.closedForm}
\end{align}
Owing to Assumption~\ref{ass:lambda_mu_bounded}, it results $m(t)<+\infty$ for all $t>0$. Moreover, if $\lambda,\,\mu\in L^1(\R_+)$ then, letting
\begin{equation}
    \Lambda_\infty:=\int_0^{+\infty}\lambda(t)\,dt, \qquad M_\infty:=\int_0^{+\infty}\mu(t)\,dt,
    \label{eq:Lambda_M.inf}
\end{equation}
we obtain
\begin{equation}
    m(t)\xrightarrow{t\to+\infty}m_\infty:=e^{M_\infty}\int_0^{+\infty}e^{-M(s)}\lambda(s)\,ds.
    \label{eq:m_inf}
\end{equation}
Since $e^{-M_\infty}<e^{-M(s)}\leq 1$ for all $s>0$, we observe that
\begin{equation}
    \Lambda_\infty\leq m_\infty\leq e^{M_\infty}\Lambda_\infty,
    \label{eq:m_inf.bounds}
\end{equation}
thus $m_\infty$ is finite and in general non-zero.

Let now
$$ E(t):=\sum_{n=0}^\infty n^2f(n,t) $$
be the \textit{energy} of the process $\{N_t,\,t\geq 0\}$ at time $t$. With $\varphi(i)=i^2$ in~\eqref{eq:Boltz_f.weak} we get that its evolution is ruled by the equation
\begin{align}
\label{eq:E.ode}
\frac{dE}{dt}=2\mu(t)E+\bigl(2\lambda(t)+\mu(t)\bigr)m+\lambda(t)    
\end{align}
with initial condition $E(0)=0$, cf.~\eqref{eq:f.init_cond}, whence
\begin{align}
\label{eq:E.closedForm}
E(t)=e^{2M(t)}\int_0^te^{-2M(s)}\bigl[\bigl(2\lambda(s)+\mu(s)\bigr)m(s)+\lambda(s)\bigr]\,ds.    
\end{align}

\begin{remark} \label{rem:energy}
From~\eqref{eq:E.ode}, owing to the non-negativity of $\lambda$, $\mu$, and $m$, we see that
$$ \frac{dE}{dt}\geq 2\mu(t)E, $$
therefore the self-excitation intensity $\mu$ estimates from below the rate of accumulation of the energy of the process $\{N_t,\,t\geq 0\}$. From this inequality we immediately see that $E$ is non-decreasing in time (notice that $E$ is non-negative by definition), which confirms the ``accumulation-without-release'' property mentioned in Remark~\ref{rem:physics.lambda*}.
\end{remark}

If we consider again the case $\lambda,\,\mu\in L^1(\R_+)$ we may estimate
\begin{align}
    \begin{aligned}[t]
        E(t)\xrightarrow{t\to+\infty}E_\infty &:= e^{2M_\infty}\int_0^{+\infty}e^{-2M(s)}\bigl[\bigl(2\lambda(s)+\mu(s)\bigr)m(s)+\lambda(s)\bigr]\,ds \\
        &\leq e^{2M_\infty}[(2\Lambda_\infty+M_\infty)m_\infty+\Lambda_\infty] \\
        &\leq e^{2M_\infty}[(2\Lambda_\infty+M_\infty)e^{M_\infty}+1]\Lambda_\infty,
    \end{aligned}
    \label{eq:E_inf}
\end{align}
where we have used that $m(t)\leq m_\infty$ for all $t>0$ (because $m(t)$ is non-decreasing, indeed from~\eqref{eq:m.ODE} it results $\frac{dm}{dt}\geq 0$) along with the upper bound in~\eqref{eq:m_inf.bounds}. Therefore, the energy is asymptotically finite. Moreover, plugging~\eqref{eq:m.closedForm},~\eqref{eq:m_inf} into~\eqref{eq:E_inf} and using again repeatedly $e^{-M_\infty}<e^{-M(s)}$ for all $s\geq 0$ we obtain
\begin{align}
    \resizebox{.9\linewidth}{!}{$
    \begin{aligned}[t]
        E_\infty &\geq e^{2M_\infty}\int_0^{+\infty}e^{-2M(s)}\bigl(2\lambda(s)+\mu(s)\bigr)\left(e^{M(s)}
            \int_0^s e^{-M(r)}\lambda(r)\,dr\right)ds+e^{M_\infty}m_\infty \\
        &\geq \int_0^{+\infty}(2\lambda(s)+\mu(s))\Lambda(s)\,ds+e^{M_\infty}\Lambda_\infty\geq\Lambda_\infty^2+e^{M_\infty}\Lambda_\infty.
    \end{aligned}
    $}
    \label{eq:E_inf.low_bound}
\end{align}

The finiteness of both $m_\infty$ and $E_\infty$ implies that the \textit{internal energy}, viz. the \textit{variance}, of the distribution $f$, defined as $\sigma_\infty^2:=E_\infty-m_\infty^2$, is asymptotically finite. In addition to this,~\eqref{eq:m_inf.bounds} and~\eqref{eq:E_inf.low_bound} entail
\begin{equation}
    \sigma_\infty^2\geq\left(\Lambda_\infty+e^{M_\infty}-e^{2M_\infty}\Lambda_\infty\right)\Lambda_\infty.
    \label{eq:sigma_inf^2.low_bound}
\end{equation}
Since $0<\sigma_\infty^2<+\infty$, the asymptotic profile towards which $f$ may converge in time neither shrinks on a single value of $n$ (the case $\sigma_\infty^2=0$) nor spreads on the whole $\N$ (the case $\sigma_\infty^2=+\infty$). Therefore, the process $\{N_t,\,t\geq 0\}$ may display interesting asymptotics.

From~\eqref{eq:sigma_inf^2.low_bound}, assuming $\Lambda_\infty>0$ (which amounts to excluding the case $\lambda\equiv 0$), we see that a sufficient condition for $\sigma_\infty^2>0$ is $(1-e^{2M_\infty})\Lambda_\infty+e^{M_\infty}>0$, namely $\Lambda_\infty<\frac{e^{M_\infty}}{e^{2M_\infty}-1}$. Hence, for fixed $M_\infty\geq 0$, if $\Lambda_\infty$ is sufficiently small a non-vanishing asymptotic internal energy may be guaranteed. Notice that if $M_\infty=0$, i.e. if $\mu\equiv 0$ (no self-excitation), then $\sigma_\infty^2>0$ for every $\Lambda_\infty\geq 0$.

\subsection{Estimates on~\texorpdfstring{$\boldsymbol{f}$}{} and characterisation of the tail}
\label{sec:f.bounds}
Besides~\eqref{eq:f(0,t)}, from the recursive relationship~\eqref{eq:f.recurrence} it is in general hard to obtain an explicit expression of $f(n,t)$ for $n\geq 1$ and $t>0$. Nonetheless, we can establish a useful estimate, which in some regimes depicts quite accurately the trend of $f$.
\begin{theorem} \label{theo:f.bound}
It results
\begin{equation}
    f(n,t)\leq\frac{e^{-\Lambda(t)}}{n!}\bigl(\Lambda(t)+M(t)(n-1)\bigr)^n
    \label{eq:f.bound}
\end{equation}
for every $n\in\N$, $n\geq 1$, and every $t>0$.
\end{theorem}
\begin{proof}
We proceed by induction on $n$.

Using~\eqref{eq:f(0,t)}, from~\eqref{eq:f.recurrence} we compute explicitly
$$ f(1,t)=e^{-(\Lambda(t)+M(t))}\int_0^t e^{M(s)}\lambda(s)\,ds\leq e^{-\Lambda(t)}\int_0^t\lambda(s)\,ds=\Lambda(t)e^{-\Lambda(t)}, $$
where we have used that $M$ is non-decreasing, thus $M(s)\leq M(t)$ for all $s\leq t$. This proves that~\eqref{eq:f.bound} holds for $n=1$.

We assume now that~\eqref{eq:f.bound} holds for a certain $n\geq 1$ and test it for $n+1$:
\begin{align*}
    f(n+1,t) &= e^{-(\Lambda(t)+M(t)(n+1))}\int_0^te^{\Lambda(s)+M(s)(n+1)}(\lambda(s)+\mu(s)n)f(n,s)\,ds \\
    &\leq \frac{e^{-\Lambda(t)}}{n!}\int_0^t(\lambda(s)+\mu(s)n)\bigl(\Lambda(s)+M(s)n\bigr)^n\,ds \\
    &= \frac{e^{-\Lambda(t)}}{n!}\int_0^t\frac{1}{n+1}\cdot\frac{d}{ds}\bigl(\Lambda(s)+M(s)n\bigr)^{n+1}\,ds \\
    &= \frac{e^{-\Lambda(t)}}{(n+1)!}\bigl(\Lambda(t)+M(t)n\bigr)^{n+1}.
\end{align*}

We conclude that~\eqref{eq:f.bound} holds for every $n\in\N$, $n\geq 1$, and every $t>0$.
\end{proof}

From Theorem~\ref{theo:f.bound} we may characterise the tail of the distribution $f(\cdot,t)$ at every $t>0$, i.e. the trend of $f(n,t)$ for $n$ large (and $t$ fixed). By Stirling's formula we have $n!\sim\sqrt{2\pi n}\left(\frac{n}{e}\right)^n$ when $n\to\infty$, whence
\begin{align}
    f(n,t)\leq\frac{e^{-\Lambda(t)}}{n!}\bigl(\Lambda(t)+M(t)(n-1)\bigr)^n\sim\frac{e^{-\Lambda(t)}}{\sqrt{2\pi}}\cdot\frac{(eM(t))^n}{\sqrt{n}}
    \label{eq:bound.f.Stirling}
\end{align} 
for $n$ sufficiently large. Since $M(t)$ is non-decreasing and $M(0)=0$, the equation $eM(t)=1$ may admit a solution. If this is the case, we denote $t^\ast:=\min\{t>0\,:\,eM(t)=1\}$. Then, for every $t<t^\ast$ the distribution $f(\cdot,t)$ decays quicker than exponentially when $n\to\infty$, thereby exhibiting a \textit{slim} tail. For $t=t^\ast$, it decays at least like $\frac{1}{\sqrt{n}}$. Finally, for $t>t^\ast$, if $M(t)$ is not constantly equal to $\frac{1}{e}$ the asymptotic estimate~\eqref{eq:bound.f.Stirling} does not allow us to predict the shape of the tail for $n$ large, because $\frac{(eM(t))^n}{\sqrt{n}}\to+\infty$ when $n\to\infty$. Conversely, if $M(t)<\frac{1}{e}$ for every $t\geq 0$ then $t^\ast$ does not exist and at all times $f(\cdot,t)$ features a slim tail analytically estimated by~\eqref{eq:bound.f.Stirling}.

On the whole, the ultimate characterisation of the tail of $f(\cdot,t)$ may be obtained as a consequence of the following fact:
\begin{theorem} \label{theo:moments_bounded}
Assume $\lambda,\,\mu\in L^1(\R_+)$. Then $f$ has statistical moments of any order uniformly bounded in time.
\end{theorem}
\begin{proof}
Let
$$ \cM_k(t):=\sum_{i=0}^{\infty}i^kf(i,t), \qquad k\in\N, $$
be the $k$-th statistical moment\footnote{With reference to Section~\ref{sect:basic_prop.f}, we have for instance $\cM_1=m$, $\cM_2=E$.} of $f$ at time $t$. Taking $\varphi(i)=i^k$ in~\eqref{eq:Boltz_f.weak} we obtain that $\cM_k$ obeys the equation
$$ \frac{d\cM_k}{dt}=\sum_{i=0}^\infty(\lambda(t)+\mu(t)i)\left((i+1)^k-i^k\right)f(i,t). $$
Writing $(i+1)^k=\sum_{j=0}^k\binom{k}{j}i^j=\sum_{j=0}^{k-1}\binom{k}{j}i^j+i^k$ we further obtain
\begin{align*}
    \frac{d\cM_k}{dt} &= \sum_{j=0}^{k-1}\binom{k}{j}\sum_{i=0}^\infty\left(\lambda(t)i^j+\mu(t)i^{j+1}\right)f(i,t)
        =\sum_{j=0}^{k-1}\binom{k}{j}\left(\lambda(t)\cM_j+\mu(t)\cM_{j+1}\right) \\
    &= \sum_{j=0}^{k-2}\binom{k}{j}\left(\lambda(t)\cM_j+\mu(t)\cM_{j+1}\right)+k\lambda(t)\cM_{k-1}+k\mu(t)\cM_k.
\end{align*}
Multiplying both sides by $e^{-kM(t)}$ yields
$$ \frac{d}{dt}\left(e^{-kM(t)}\cM_k\right)=
    e^{-kM(t)}\left\{\sum_{j=0}^{k-2}\binom{k}{j}\left(\lambda(t)\cM_j+\mu(t)\cM_{j+1}\right)+k\lambda(t)\cM_{k-1}\right\}, $$
whence, integrating on $[0,\,t]$, $t>0$,
\begin{multline*}
    \cM_k(t)=e^{kM(t)}\cM_k(0)+\int_0^te^{k(M(t)-M(s))} \\
    \times\left\{\sum_{j=0}^{k-2}\binom{k}{j}\bigl(\lambda(s)\cM_j(s)+\mu(s)\cM_{j+1}(s)\bigr)+k\lambda(s)\cM_{k-1}(s)\right\}ds.
\end{multline*}
In particular, $\cM_0(0)=1$ whereas $\cM_k(0)=0$ for all $k>0$ because $f(n,0)=\delta_0(n)$.

Clearly, $\cM_0\equiv 1$ is uniformly bounded. Assume now that the first $k$ moments of $f$ are uniformly bounded in time, i.e. that there exist constants $C_j>0$, $j=0,\,\dots,\,k$, such that $\cM_j(t)\leq C_j$ for all $t\geq 0$. Then also $\cM_{k+1}$ is uniformly bounded in time, indeed:
\begin{align*}
    \cM_{k+1}(t) &= \int_0^te^{(k+1)(M(t)-M(s))}
        \Biggl\{\sum_{j=0}^{k-1}\binom{k+1}{j}\bigl(\lambda(s)\cM_j(s)+\mu(s)\cM_{j+1}(s)\bigr) \\
    &\phantom{=\int_0^te^{(k+1)(M(t)-M(s))}\Biggl\{} +(k+1)\lambda(s)\cM_k(s)\Biggr\}\,ds,
\intertext{where we have taken into account that $\cM_{k+1}(0)=0$ for every $k\geq 0$, and further}
    &\leq e^{(k+1)M_\infty}\int_0^{+\infty}\left\{\sum_{j=0}^{k-1}\binom{k+1}{j}\left(\lambda(s)C_j+\mu(s)C_{j+1}\right)+(k+1)\lambda(s)C_k\right\}ds \\
    &= e^{(k+1)M_\infty}\sum_{j=0}^{k-1}\binom{k+1}{j}\left(\Lambda_\infty C_j+M_\infty C_{j+1}\right)+(k+1)\Lambda_\infty C_k=:C_{k+1},
\end{align*}
in view of~\eqref{eq:Lambda_M.inf} along with $M(t)\leq M_\infty$ for all $t\geq 0$ because $M$ is non-decreasing.

By induction on $k$, we conclude that all moments of $f$ are uniformly bounded in time.
\end{proof}

Theorem~\ref{theo:moments_bounded} implies straightforwardly:
\begin{corollary}
\label{cor:slimTails}
If $\lambda,\,\mu\in L^1(\R_+)$ then $f(\cdot,t)$ has a slim tail for all $t>0$.
\end{corollary}
\begin{proof}
Indeed, $f(\cdot,t)$ has finite statistical moments of any order for all $t>0$.
\end{proof}

The slim tail of $f(\cdot,t)$ indicates that, at every time, the occurrence of a large total number of accidents up to that time is a quite unlikely event. Indeed, for $\bar{n}\in\N$ we have $\Prob{N_t\geq\bar{n}}=\sum_{n=\bar{n}}^\infty f(n,t)$ and the larger $\bar{n}$ the smaller this quantity owing to the slimness of the tail of $f$.

\subsection{The case of~\texorpdfstring{$\boldsymbol{\lambda}$}{},~\texorpdfstring{$\boldsymbol{\mu}$}{} constant}
When $\lambda,\,\mu>0$ are constant the intensity $\lambda^\ast$ takes the special form~\eqref{eq:Hawkes.mu_const}, which allows for further developments of the theory. We refer the interested reader to Section~\ref{sect:lambda_mu_const.f} of the Supplementary Material.

\section{Accident times}
\label{sect:acc_times}
After investigating the distribution function of the number of accidents $N_t$, we turn our attention to the law of the jump times $\{T_k\}_{k\geq 1}$, which in this context we rename \textit{accident times}. Next, we will investigate also the law of the holding times $\{\Delta{T}_k\}_{k\geq 1}$, that here we call \textit{interaccident times} as they represent the time elapsing between an accident and the subsequent one. See Figure~\ref{fig:Nt}. We will ascertain that, due to the self-excitation of the process $\{N_t,\,t\geq 0\}$, the interaccident times are not identically distributed, hence that $\{N_t,\,t\geq 0\}$ is indeed not a classical renewal process. Despite this technical difficulty, we will provide analytical results which characterise quite explicitly the (probabilistic) occurrence of car accidents in time.

\subsection{Law of the accident times}
Formally, we may define the time of occurrence of the $k$-th accident, $k\in\N$, $k\geq 1$, as
$$ T_k:=\inf\{t\geq 0\,:\,N_t=k\}, $$
whence we deduce that the event $\{T_k\leq t\}$ is the same as $\{N_t\geq k\}$. Therefore:
$$ \Prob{T_k\leq t}=\Prob{N_t\geq k}=1-\sum_{n=0}^{k-1}\Prob{N_t=n}=1-\sum_{n=0}^{k-1}f(n,t). $$

Letting $g_k=g_k(t)$ be the probability density function (pdf) of the random variable $T_k$, we may write $g_k(t)=\frac{d}{dt}\Prob{T_k\leq t}=-\sum_{n=0}^{k-1}\partial_tf(n,t)$ and further, expressing $\partial_tf(n,t)$ via~\eqref{eq:Boltz_f.strong},
$$ g_k(t)=\lambda(t)\sum_{n=0}^{k-1}\bigl(f(n,t)-f(n-1,t)\bigr)+\mu(t)\sum_{n=0}^{k-1}\bigl(nf(n,t)-(n-1)f(n-1,t)\bigr). $$
These telescopic sums may be evaluated explicitly. Taking also Remark~\ref{rem:f.n_neg} into account we get finally
\begin{equation}
    g_k(t)=\bigl(\lambda(t)+\mu(t)(k-1)\bigr)f(k-1,t).
    \label{eq:gk}
\end{equation}

It is instructive to check explicitly that such a $g_k$ is a pdf:
\begin{proposition} \label{prop:gk.pdf}
Assume $\lambda\not\in L^1(\R_+)$ and either $\mu\in L^1(\R_+)$ or $\mu$ constant. Then $g_k$ given by~\eqref{eq:gk} is a pdf in $\R_+$ for every $k\in\N$, $k\geq 1$.
\end{proposition}
\begin{proof}
Clearly, $g_k(t)\geq 0$ for all $t\geq 0$. Thus, the key point is to show that it has unitary integral in $\R_+$.

With $k=1$ we have explicitly $g_1(t)=\lambda(t)f(0,t)=\lambda(t)e^{-\Lambda(t)}$ in view of~\eqref{eq:f(0,t)} and we may compute $\int_0^{+\infty}g_1(t)\,dt=1-e^{-\Lambda_\infty}=1$, since $\Lambda_\infty=+\infty$ owing to the assumption $\lambda\not\in L^1(\R_+)$.

We now assume that $g_k$ has unitary integral in $\R_+$ for a certain $k\in\N$ and we check the analogous property for $g_{k+1}$. First, combining~\eqref{eq:Boltz_f.strong} and~\eqref{eq:gk} we notice that $\partial_tf(k,t)=g_k(t)-g_{k+1}(t)$, whence, integrating in time,
$$ \lim_{t\to +\infty}f(k,t)-f(k,0)=1-\int_0^{+\infty}g_{k+1}(t)\,dt. $$
Because of~\eqref{eq:f(0,t)} it results $f(k,0)=0$ for every $k\geq 1$ whereas~\eqref{eq:f.bound} together with the present assumptions on $\lambda,\,\mu$ imply $\lim\limits_{t\to +\infty}f(k,t)=0$. Hence $\int_0^{+\infty}g_{k+1}(t)\,dt=1$ and the thesis follows by induction on $k$.
\end{proof}

\begin{remark}
The assumptions of Proposition~\ref{prop:gk.pdf} are sufficiently representative of the cases of interest in this work but they are not the most general possible. For instance, if $\lambda$ is non-integrable then $\mu$ may be in turn a non-integrable function, not necessarily constant, such that $e^{-\frac{\Lambda(t)}{n}}M(t)\to 0$ when $t\to +\infty$, in such a way that from~\eqref{eq:f.bound} one still has $f(k,t)\to 0$ when $t\to +\infty$.

On the other hand, if $\lambda$ is integrable then $\Lambda_\infty<+\infty$ and $\int_0^{+\infty}g_1(t)\,dt=1-e^{-\Lambda_\infty}<1$. In this case, $g_1$, and likewise the other $g_k$'s, is not a pdf because there is a non-zero probability that $T_1=+\infty$, i.e. that the first accident never occurs (notice that $\int_0^{+\infty}g_1(t)\,dt=\Prob{T_1<+\infty}$). Such a probability can be computed as
$$ \Prob{T_1=+\infty}=\Prob{N_t=0\ \text{for arbitrarily large } t}=\lim_{t\to +\infty}f(0,t)=e^{-\Lambda_\infty}, $$
i.e. precisely the value to be added to $\Prob{T_1<+\infty}$ to obtain $\Prob{T_1\leq +\infty}=1$.
\end{remark}

\subsection{Law of the interaccident times}
\label{sect:LawIntermediate}
The discussion set forth about the law of the $T_k$'s is propaedeutical to the investigation of the interaccident times
$$ \Delta{T}_k:=T_k-T_{k-1}, \qquad k\in\N,\ k\geq 1. $$
As already mentioned, $\Delta{T}_k$ expresses the time elapsing between the occurrence of the $(k-1)$-th and the $k$-th accident or, in other words, the waiting time of the $k$-th accident after the $(k-1)$-th has occurred.

For the sake of convenience, throughout this section we assume that $\lambda,\,\mu$ are such that the $g_k$'s are the pdf's of the $T_k$'s as stated by Proposition~\ref{prop:gk.pdf}. Letting conventionally $T_0=0$, we observe that $\Delta{T}_1\equiv T_1$, therefore 
\begin{align}
\label{eq:DeltaT1}
 \Delta{T}_1\sim g_1(t)=\lambda(t)e^{-\Lambda(t)}   
\end{align}
cf.~\eqref{eq:f(0,t)} and~\eqref{eq:gk}. In particular, the law of $\Delta{T}_1$ is unaffected by $\mu$, which is reasonable considering that the self-excitation of the accident process may play a role only after the very first accident has occurred, hence not in the time interval $[0,\,T_1)$.

The determination of the laws of the $\Delta{T}_k$'s for $k>1$ is much less straightforward, due to the statistical dependence among the $T_k$'s. Notice indeed that $T_{k-1},\,T_k$ are clearly not independent because e.g., $T_{k-1}\leq T_k$.

To obtain the law of $\Delta{T}_k$ we proceed by determining first the joint law of the pair $(T_{k-1},\,T_k)$, whose pdf we denote by $g_{k-1,k}=g_{k-1,k}(s,t)$ with $(s,\,t)\in\R_+^2$. Let $g_{T_k\vert T_{k-1}}=g_{T_k\vert T_{k-1}}(t\vert s)$ be the pdf of $T_k$ conditioned to $T_{k-1}$. Then
\begin{equation}
    g_{k-1,k}(s,t)=g_{T_k\vert T_{k-1}}(t\vert s)g_{k-1}(s).
    \label{eq:gk-1,k.def}
\end{equation}
Since $g_{k-1}$ is provided by~\eqref{eq:gk}, the only quantity to be really found is $g_{T_k\vert T_{k-1}}$. First, we consider that $g_{T_k\vert T_{k-1}}(t\vert s)=\partial_t\Prob{T_k\leq t\vert T_{k-1}=s}$. Consequently, we compute:
$$ \Prob{T_k\leq t\vert T_{k-1}=s}=\Prob{N_t\geq k\vert T_{k-1}=s}=1-\sum_{n=0}^{k-1}\Prob{N_t=n\vert T_{k-1}=s}. $$
Clearly, if $t\leq s$ then $\Prob{T_k\leq t\vert T_{k-1}=s}=0$, therefore we may focus on the case $t>s$ only. On the other hand, for $t>s$ it results $\Prob{N_t=n\vert T_{k-1}=s}=0$ for every $n<k-1$, hence finally
\begin{equation}
    \Prob{T_k\leq t\vert T_{k-1}=s}=1-\Prob{N_t=k-1\vert T_{k-1}=s}.
    \label{eq:Tk|Tk-1.intermediate_formula}
\end{equation}
This formula shows that a key quantity in our computation is the law of $N_t$ conditioned to $T_{k-1}$. Therefore, we make now a small detour to tackle the problem of determining $\Prob{N_t=n\vert T_k=s}$ for $k,\,n\in\N$ and $0\leq s<t$.

Equation~\eqref{eq:Boltz_f.strong}, together with the initial condition~\eqref{eq:f.init_cond}, provides $\Prob{N_t=n}$ for every $n\in\N$ and $t>0$. If, upon fixing $s>0$ and $k\in\N$, we replace~\eqref{eq:f.init_cond} with the condition
\begin{equation}
    f(n,s)=
    \begin{cases}
        1 & \text{if } n=k \\
        0 & \text{if } n\neq k
    \end{cases}
    \label{eq:f.init_cond.s}
\end{equation}
we may interpret the solution to~\eqref{eq:Boltz_f.strong} for $t\in (s,\,+\infty)$ precisely as $\Prob{N_t=n\vert T_k=s}$. Indeed, the new initial condition~\eqref{eq:f.init_cond.s} establishes that the evolution of the system begins at time $t=s$ when the $k$-th accident occurs, i.e. that $T_k=s$. In order to stress this interpretation, we introduce the notation
$$ f_{T_k=s}(n,t):=\Prob{N_t=n\vert T_k=s}, $$
which solves~\eqref{eq:Boltz_f.strong} for $t>s$ with initial condition~\eqref{eq:f.init_cond.s}. Clearly, $f_{T_k=s}(n,t)=0$ for $t>s$ and $n<k$.

In~\eqref{eq:Tk|Tk-1.intermediate_formula} we need in particular $f_{T_{k-1}=s}(k-1,t)$, which from~\eqref{eq:Boltz_f.strong} with $n=k-1$ and  $f_{T_{k-1}=s}(k-2,t)=0$ solves
\begin{equation*}
    \begin{cases}
        \partial_tf_{T_{k-1}=s}(k-1,t)=-\bigl(\lambda(t)+\mu(t)(k-1)\bigr)f_{T_{k-1}=s}(k-1,t), & t>s \\[2mm]
        f_{T_{k-1}=s}(k-1,s)=1,
    \end{cases}
\end{equation*}
thus $f_{T_{k-1}=s}(k-1,t)=e^{-(\Lambda(t)-\Lambda(s))-(M(t)-M(s))(k-1)}$ for $t>s$ while $f_{T_{k-1}=s}(k-1,t)=0$ for $t<s$. Consequently,
\begin{align*}
    g_{T_k\vert T_{k-1}}(t\vert s) &= \partial_t\Prob{T_k\leq t\vert T_{k-1}=s}, \\
\intertext{which, owing to~\eqref{eq:Tk|Tk-1.intermediate_formula}, becomes}
    &= -\partial_t\Prob{N_t=k-1\vert T_{k-1}=s}=-\partial_tf_{T_{k-1}=s}(k-1,t) \\
    &= \begin{cases}
            \bigl(\lambda(t)+\mu(t)(k-1)\bigr)e^{-(\Lambda(t)-\Lambda(s))-(M(t)-M(s))(k-1)} & \text{if } t>s \\
            0 & \text{if } t<s
        \end{cases}
\end{align*}
and finally, recalling~\eqref{eq:gk-1,k.def} and invoking~\eqref{eq:gk} to express $g_{k-1}(s)$,
\begin{equation}
    g_{k-1,k}(s,t)=
    \begin{cases}
        \bigl(\lambda(s)+\mu(s)(k-2)\bigr)\bigl(\lambda(t)+\mu(t)(k-1)\bigr) \\
        \quad\times e^{-(\Lambda(t)-\Lambda(s))-(M(t)-M(s))(k-1)}f(k-2,s) & \text{if } t>s \\
        0 & \text{if } t<s.
    \end{cases}
    \label{eq:gk-1,k}
\end{equation}
We remark that here $f(k-2,s)$ is the solution to~\eqref{eq:Boltz_f.strong} for $n=k-2$ with initial condition~\eqref{eq:f.init_cond}.

The following check is in order:
\begin{lemma}
Assume $\lambda\not\in L^1(\R_+)$. Then~\eqref{eq:gk-1,k} defines a pdf in $\R_+^2$ for every $k\in\N$, $k\geq 2$.
\label{lem:gIsPdf}
\end{lemma}
\begin{proof}
Clearly, $g_{k-1,k}(s,t)\geq 0$ for all $s,\,t\in\R_+$, therefore we only have to prove that $g_{k-1,k}$ has unitary integral in $\R_+^2$:
\begin{align*}
    \int_0^{+\infty}\int_0^{+\infty}g_{k-1,k}(s,t)\,ds\,dt &= \int_0^{+\infty}\bigl(\lambda(s)+\mu(s)(k-2)\bigr)e^{\Lambda(s)+M(s)(k-1)}f(k-2,s) \\
    &\phantom{=} \times\left(\int_s^{+\infty}\bigl(\lambda(t)+\mu(t)(k-1)\bigr)e^{-\Lambda(t)-M(t)(k-1)}\,dt\right)ds \\
    &= \int_0^{+\infty}\bigl(\lambda(s)+\mu(s)(k-2)\bigr)e^{\Lambda(s)+M(s)(k-1)}f(k-2,s) \\
    &\phantom{=} \times\left(-e^{-\Lambda(t)-M(t)(k-1)}\right\vert_s^{+\infty}ds; \\
\intertext{considering that $e^{-\Lambda(t)-M(t)(k-1)}\leq e^{-\Lambda(t)}\to 0$ for $t\to +\infty$ because, by assumption, $\Lambda(t)\to +\infty$ when $t\to +\infty$ and recalling~\eqref{eq:gk} we obtain finally}
    &= \int_0^{+\infty}g_{k-1}(s)\,ds
\end{align*}
and the thesis follows from Proposition~\ref{prop:gk.pdf}.
\end{proof}

Let now $\tau\geq 0$. The joint pdf $g_{k-1,k}$ allows us to obtain the law of the interaccident time $\Delta{T}_k$ as:
\begin{align*}
    \Prob{\Delta{T}_k\leq\tau} &= \Prob{T_k\leq T_{k-1}+\tau}=\int_0^{+\infty}\int_0^{s+\tau}g_{k-1,k}(s,t)\,dt\,ds \\
    &= \int_0^{+\infty}\bigl(\lambda(s)+\mu(s)(k-2)\bigr)e^{\Lambda(s)+M(s)(k-1)}f(k-2,s) \\
    &\phantom{=} \quad\times\left(\int_s^{s+\tau}\bigl(\lambda(t)+\mu(t)(k-1)\bigr)e^{-\Lambda(t)-M(t)(k-1)}\,dt\right)ds \\
    &= -\int_0^{+\infty}\bigl(\lambda(s)+\mu(s)(k-2)\bigr)e^{-(\Lambda(s+\tau)-\Lambda(s))-(M(s+\tau)-M(s))(k-1)} \\
    &\phantom{=} \qquad\qquad\times f(k-2,s)\,ds \\
    &\phantom{=} +\int_0^{+\infty}\bigl(\lambda(s)+\mu(s)(k-2)\bigr)f(k-2,s)\,ds \\
    &= 1-\int_0^{+\infty}\bigl(\lambda(s)+\mu(s)(k-2)\bigr)e^{-(\Lambda(s+\tau)-\Lambda(s))-(M(s+\tau)-M(s))(k-1)} \\
    &\phantom{=} \qquad\qquad\times f(k-2,s)\,ds,
\end{align*}
whence, denoting by $h_k$ the pdf of $\Delta{T}_k$, we compute finally
\begin{align}
    \begin{aligned}[t]
        h_k(\tau) &= \frac{d}{d\tau}\Prob{\Delta{T}_k\leq\tau} \\
        &= \int_0^{+\infty}\bigl(\lambda(s)+\mu(s)(k-2)\bigr)\bigl(\lambda(s+\tau)+\mu(s+\tau)(k-1)\bigr) \\
        &\phantom{=} \quad\times e^{-(\Lambda(s+\tau)-\Lambda(s))-(M(s+\tau)-M(s))(k-1)}f(k-2,s)\,ds,
    \end{aligned}
    \label{eq:hk}
\end{align}
for $\tau\geq 0$ and $k\in\N$, $k\geq 1$. Since we do not have, in general, a closed expression of $f(k-2,s)$, we cannot be more explicit than this in the expression of $h_k$. Notice, however, that since in general $h_k$ depends on $k$ the $\Delta{T}_k$'s are not identically distributed.

\begin{remark} \label{rem:hk.pdf}
Equation~\eqref{eq:hk} shows that $h_k$ is non-negative. Moreover, under the assumption $\lambda\not\in L^1(\R_+)$, cf. Lemma~\ref{lem:gIsPdf}, we have $\Prob{\Delta{T}_k< +\infty}=1$. Then $\int_0^{+\infty} h_k(\tau)\,d\tau=\Prob{\Delta{T}_k< +\infty}-\Prob{\Delta{T}_k\leq 0}=1-\Prob{\Delta{T}_k\leq 0}$ and  we further observe that $\Prob{\Delta{T}_k\leq 0} =\Prob{\Delta{T}_k=0}$, because interaccident times are non-negative by construction. Also, $\Delta{T}_k$ cannot be zero for $T_{k-1}$, $T_k$ are distinct accident times by definition. Hence finally $\Prob{\Delta{T}_k\leq 0}=0$ and $h_k$ is indeed a pdf.
\end{remark}

\subsection{Aggregate interaccident time and its law}
\label{sect:aggregateInterTimes}
A further quantity of interest, which is often studied also experimentally, is what we might call the \textit{aggregate interaccident time}. This quantity is meant to represent the size of a generic interaccident time, regardless of whether it is the first, second, \dots, $k$-th interaccident time in a series of accidents up to a certain time horizon, viz. final time, $\cT>0$.

We formalise mathematically this concept by introducing a new random variable $\Delta{T}_K\in\R_+$, being $K\in\N$, $K\geq 1$, in turn a random variable, which models a random sampling of the $\Delta{T}_k$'s within the time horizon $\cT>0$. For $\tau\geq 0$, the event $\{\Delta{T}_K\leq\tau\}$ refers therefore to the occurrence of any two accidents with a temporal gap at most equal to $\tau$ within the time frame $[0,\,\cT]$.

Notice that
\begin{align*}
    \Prob{\Delta{T}_K\leq\tau} &= \sum_{k=1}^\infty\Prob{\Delta{T}_K\leq\tau\vert K=k}\Prob{K=k} \\
    &= \sum_{k=1}^\infty\Prob{\Delta{T}_k\leq\tau}\Prob{K=k},
\end{align*}
therefore the pdf of $\Delta{T}_K$, say $h=h(\tau)$, is obtained as:
\begin{align}
\label{eq:def.h}
h(\tau)=\frac{d}{d\tau}\Prob{\Delta{T}_K\leq\tau}=\sum_{k=1}^\infty h_k(\tau)\Prob{K=k},    
\end{align}
namely as a weighted average of the pdf's $\{h_k\}_{k=1}^\infty$ of the random variables $\Delta{T}_k$.

The point is how to construct the random variable $K$, in particular the values $\{\Prob{K=k}\}_{k=1}^\infty$ defining its law. This amounts to modelling the process by which a generic interaccident time $\Delta{T}_k$ is randomly sampled within the time horizon $\cT$. In particular, one should take into account that the smaller $k$ the more probable the occurrence of the $(k-1)$-th and $k$-th accidents within the time frame $[0,\,\cT]$. Proportionally to $\cT$, it should then be more probable to sample interaccident times $\Delta{T}_k$ with $k$ small than with $k$ large.

This discussion leads us to envisage a model in which the probability to sample a certain interaccident time $\Delta{T}_k$, hence $K=k$, is proportional to the probability that the $k$-th accident occurs within the time horizon $\cT$:
$$ \Prob{K=k}=C\Prob{T_k\leq\cT}=C\int_0^\cT g_k(t)\,dt, $$
where $C>0$ is a proportionality constant. To find it, we impose
\begin{align*}
    1=\sum_{k=1}^\infty\Prob{K=k} &= C\int_0^\cT\sum_{k=1}^\infty g_k(t)\,dt \\
\intertext{and recalling~\eqref{eq:gk}}
    &= C\int_0^\cT\left(\lambda(t)\sum_{k=1}^\infty f(k-1,t)+\mu(t)\sum_{k=1}^\infty(k-1)f(k-1,t)\right)dt; \\
\intertext{Proposition~\ref{prop:f.prob_distr} and~\eqref{eq:m} imply then}
    &= C\int_0^\cT\bigl(\lambda(t)+\mu(t)m(t)\bigr)\,dt \\
\intertext{and furthermore, owing to~\eqref{eq:m.ODE},}
    &= C\int_0^\cT\frac{dm(t)}{dt}\,dt=C(m(\cT)-m(0))=Cm(\cT),
\end{align*}
where we have taken the initial condition~\eqref{eq:f.init_cond} into account. Thus $C=\frac{1}{m(\cT)}$, so that in conclusion we set $\Prob{K=k}:=\frac{\Prob{T_k\leq\cT}}{m(\cT)}$ and we define
\begin{align}
    \begin{aligned}[t]
        h(\tau) &= \frac{1}{m(\cT)}\sum_{k=1}^\infty h_k(\tau)\Prob{T_k\leq\cT} \\
        &= \frac{1}{m(\cT)}\sum_{k=1}^\infty h_k(\tau)\int_0^\cT\bigl(\lambda(t)+\mu(t)(k-1)\bigr)f(k-1,t)\,dt
    \end{aligned}
    \label{eq:h}
\end{align}
for $\tau\geq 0$.

\begin{remark}
Since under $\lambda\not\in L^1(\R_+)$ the $h_k$'s are pdfs, cf. Remark~\ref{rem:hk.pdf}, and in view of the construction of the law of $K$, also $h$ turns out to be a pdf straightforwardly.
\end{remark}

\subsection{The case of~\texorpdfstring{$\boldsymbol{\lambda}$}{},~\texorpdfstring{$\boldsymbol{\mu}$}{} constant}
To detail further the theory of accident and interaccident times when $\lambda$, $\mu$ are positive constants, we refer the interested reader to Section~\ref{sect:aggregateInterTimes.constant_lambda_mu} of the Supplementary Material.

\section{Calibration of~\texorpdfstring{$\boldsymbol{\lambda}$}{} and~\texorpdfstring{$\boldsymbol{\mu}$}{} against real data}
\label{sec:realData}
Section~\ref{sect:numerics} of the Supplementary Material contains extensive numerical illustrations of the theoretical results discussed in the paper under various assumptions on $\lambda$ and $\mu$.

Here instead we go back to Figure~\ref{fig:motivation} and attempt to identify background intensity functions $\lambda$ and self-excitation intensity functions $\mu$ that fit real data. Our goal is not to select ``optimal'' (in whatever sense) $\lambda$ and $\mu$ but rather to exemplify \textit{classes} of functions that might be appropriate in the context of traffic accidents. While deferring to future research a more extensive study addressing the numerous degrees of freedom involved in the choice of $\lambda$ and $\mu$, we show that basic satisfactory examples of such classes may indeed be found.

The study focuses on a time horizon of one week, where $t=0$ is Sunday midnight and $t=10079$ is Saturday 11:59pm, corresponding to a time unit of one minute. This choice seems natural, as it is a commonly used time unit, and could be varied to a medium-size time horizon extending from some days to a few weeks. We assume $N_{0}=0$ so that accidents from the week before are neglected. The accident process $\{N_t,\,t\geq 0\}$ from our model \eqref{eq:Nt.evolution}, \eqref{eq:Prob.H}, and \eqref{eq:our_lambda*} is compared with British accident data from 2020\footnote{\url{https://data.gov.uk/dataset/cb7ae6f0-4be6-4935-9277-47e5ce24a11f/road-safety-data} -- Accessed 11th September 2024.}, where we considered all accidents within the rectangle of $[-2,1]$ longitude and $[51,52]$ latitude, which is an area west of London. There are 2205 accidents in the dataset. In the further discussion, we assume 51 weeks per year.
	
For our model with background intensity $\lambda$ from equation \eqref{eq:our_lambda*}, we choose 
\begin{align*}
	\lambda(t)= \tilde{C} c_0 \left(\frac{5}{4}+ \sin\left(\frac{\pi}{740}(t-540)\right)\right), 
\end{align*}
which models intraday fluctuations with low risks for accidents at night and larger risks in the afternoon, which have also been detected e.g.\ in \cite{Moutari.2018}. The constant $\tilde{C}\in[0,1]$ works as a weighting parameter between the background noise and the excitation accidents and can be varied. For $\tilde{C}=1$ the entire accident risk stems from the background noise and for $\tilde{C}=0$ there is no background accident risk. We choose the constant $c_0$ such that for $\tilde{C}=1$ and $\mu \equiv 0$ it holds
\begin{align*}
	\frac{2205}{51} &\approx \mathbb{E}[N_{10080}] = \int_0^{10080} \lambda(t)dt = \int_0^{10080} c_0 \left(\frac{5}{4}+ \sin\left(\frac{\pi}{740}(t-540)\right)\right)dt \\
	&\Rightarrow c_0 \approx 0.003413,
\end{align*}
representing an accident process without an excitation share that, in expectation, yields as many accidents as on average in one week in the British data. Next, we introduce an excitation share $(\tilde{C} \neq 1)$ and consider the following excitation functions
\begin{align*}
   \mu_1(t)=c_1, \qquad \mu_2(t)= c_2e^{-\frac{t}{600}}, \qquad \mu_3(t) = \frac{c_3}{t+50},
\end{align*}
with constants $c_1,c_2,c_3>0$. These functions are representatives of different integrability assumptions, i.e.\ $\mu_1 \notin L^1(\mathbb{R}_+)$, $\mu_2 \in L^1(\mathbb{R}_+)$, and $\mu_3$ close to integrability. All of them have already been considered in the previous sections, but were scaled for the particular application. We compare a scenario without excitation ($\tilde{C}=1$) with frameworks with weak excitation ($\tilde{C}=\frac{1}{2})$ and medium excitation share ($\tilde{C}=\frac{1}{3}$), for each excitation function $\mu_i, ~i=1,2,3$. In all of the three excitation functions we consider one free parameter given by either coefficient $c_1$, $c_2$, $c_3$. These have to be determined in all cases such that $\mathbb{E}[N_{10080}]\approx\frac{2205}{51}$. A Monte Carlo method on a discrete grid of possible values leads to the following choices:
\begin{align*}
	c_1^\text{weak}  = 0.00015, 
    \qquad c_2^\text{weak} =0.0081, 
    \qquad c_3^\text{weak}  = 0.6 \\ 
    c_1^\text{medium}  = 0.0002, 
    \qquad c_2^\text{medium}  = 0.0105, 
    \qquad c_3^\text{medium}  = 0.85.
\end{align*}

We start the analysis by choosing $\mu_1$, the time-constant excitation intensity function in Figure \ref{fig:realData11}. To assess whether the background intensity functions and the excitation intensity functions are suitable, we consider several illustrations. First, we compare the temporal evolutions of the number of accidents in all of the 51 weeks in 2020 that happened in the considered area (see also Figure \ref{fig:motivation}). The data is compared to the accident process from our model with medium, weak and no excitation (first row from left to right). To better compare the numbers of accidents after one week we additionally provide histograms for $N_{10080}$ in the three different cases and compare them to the total number of accidents found in the data for all weeks in the data set.

\begin{figure}[!t]
\centering
\includegraphics[width=\linewidth]{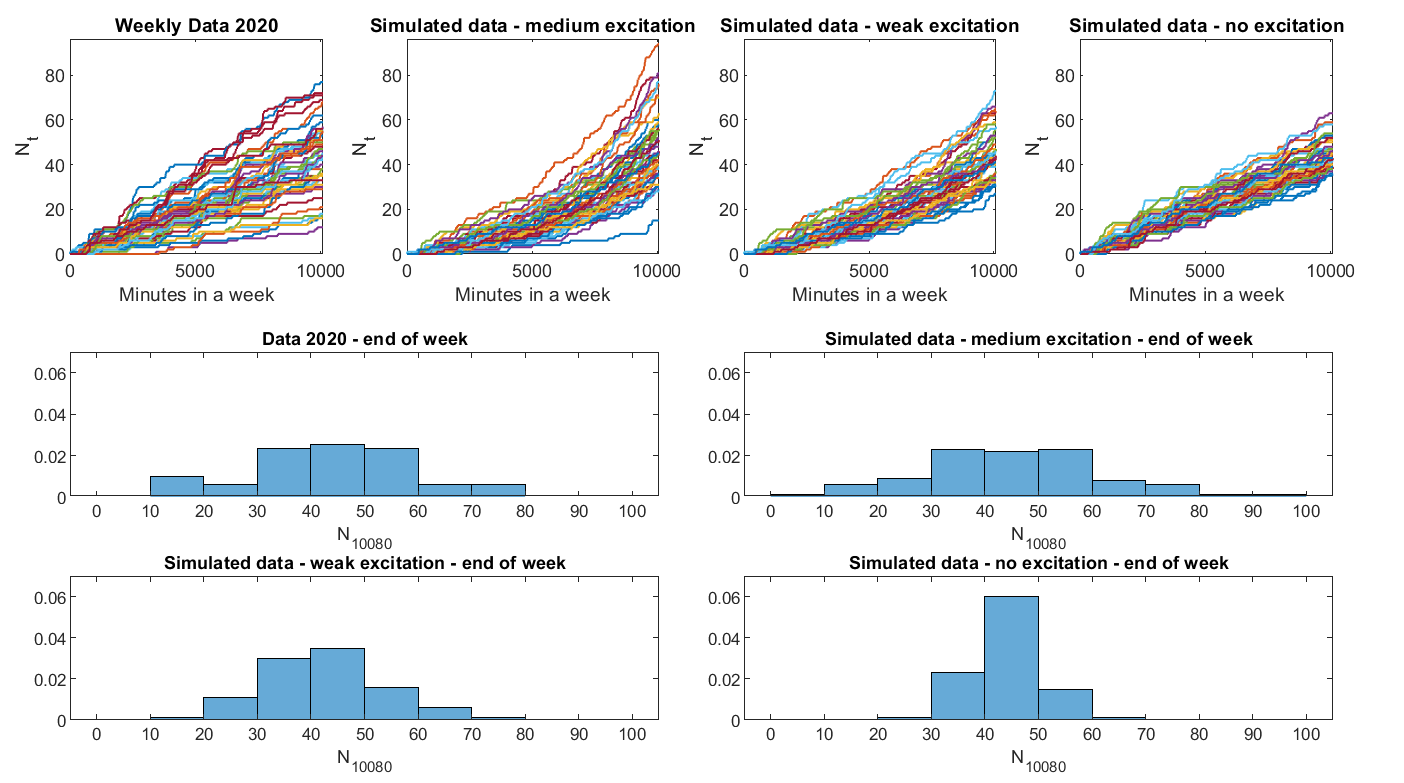}
\caption{Real data study for $\mu_1$, the time-constant excitation function. Comparison of temporal evolution of $N_{t}$ and the histograms for $N_{10080}$ with different choices of the weights to real data.}
\label{fig:realData11}
\end{figure}
			
Focusing on the temporal evolution of the number of accidents in one week, the increase of the curves given by data follows, roughly speaking, a linear pattern. The total number of accidents in a week ranges from 15 to 76 accidents per week, with the largest clustering around 45 accidents per week. Looking at the accident process from our model without excitation (top right), we also have the style of a linear increase, but we observe a significantly narrower range of $N_{10080}$, which lies between 30 and 60. This observation of over-dispersion has been mentioned in the introduction and underlines that a standard Poisson process without modification does not reproduce the real world accident process. Including the excitation shares, we see that the range widens and leads to outliers, especially for the medium excitation share. Note that on average all accident processes yield the same number of accidents per week, only their distribution and temporal evolution differ.
 
Furthermore, we observe a convex structure in the temporal evolution, representing more accidents on Fridays and Saturdays than at the beginning of the week. This effect is not present in the data and adds an unnatural property to the model when choosing $\mu_1$. Considering the histograms, the model with medium excitation share seems to be the most appropriate, since the number of accidents seems to be distributed almost symmetrically around some mean, with a minimum around 15 and a maximum around 80. As mentioned before, the model without additional excitation share is inappropriate and again underlines the necessity of our approach. However, due to the convex structure of the temporal evolution of $N_{t}$, the constant excitation function, is not an optimal choice.
	
Next, we consider $\mu_2$, which decreases exponentially in time, and present the corresponding illustration in Figure~\ref{fig:realData21}. The illustration of the evolution in time for the British data and without excitation share are the same as in Figure~\ref{fig:realData11}, except for the scaling of the vertical axis, which had to be adjusted due to the outliers in this model. Unlike before, the exponentially decaying excitation function creates more weeks with total weekly numbers of accidents around 20 and 30 and then also creates outliers, some of them resulting in values $N_{10080}>100$ which are significantly more accidents than can be found in the data. Furthermore, we observe a rather concave shaped temporal evolution of $N_{t}$ for both, the weak and medium excitation share. This effect is converse to the one obtained in Figure~\ref{fig:realData11}, but is also not represented in the British data. The histograms of $N_{10080}$ do not show a symmetric distribution, but rather a right-skewed distribution. Therefore, one can conclude that the exponentially decaying excitation function is also not appropriate for the modelling of the evolution of the number of accidents, given the British data set.
	
\begin{figure}[!t]
\centering
\includegraphics[width=\linewidth]{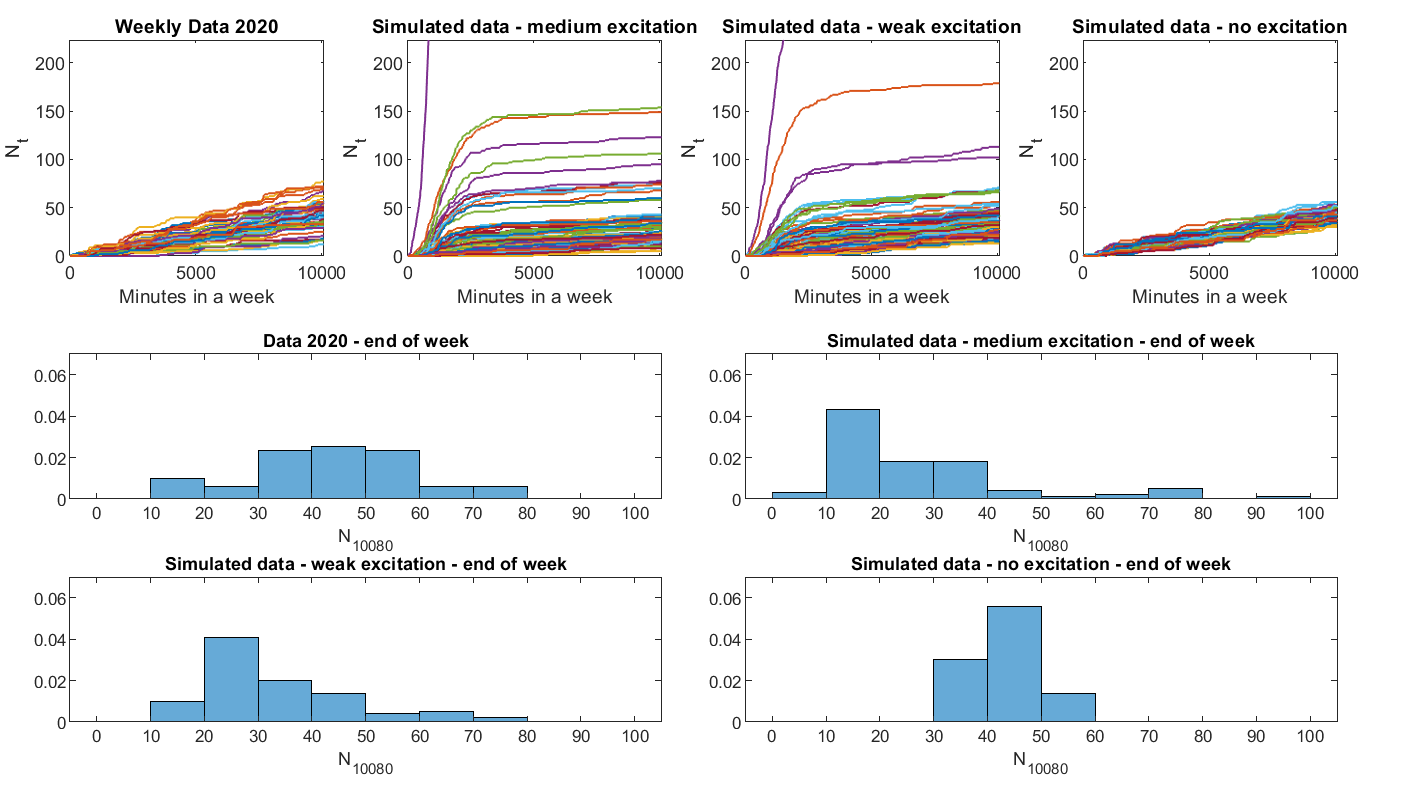}
\caption{Real data study for $\mu_2$, the exponentially decreasing excitation function. Comparison of temporal evolution of $N_{t}$ and the histograms for $N_{10080}$ with different choices of the weights to real data.}
\label{fig:realData21}
\end{figure}
	
Last, we consider $\mu_3$, which is a rational function of degree 1 in time, in Figure \ref{fig:realData31}. The upper row again shows the temporal evolution of the accident process $N_{t}$, which under these assumptions follows an approximately linear shape for each week. Increasing the share of the excitation function, the weekly total numbers lie in a wider range. For the weak share, the simulations are in good agreement with the real data. Additionally, the linear structure can also be found in real data. The histogram of the weak excitation share is not perfectly symmetric, but it is still close to the one based on the data set. Therefore, given the sinusoidally varying background intensity function, the choice of a rational function of degree 1 as the excitation intensity function seems appropriate to model a weekly accident process. 
	
\begin{figure}[!t]
\centering
\includegraphics[width=\linewidth]{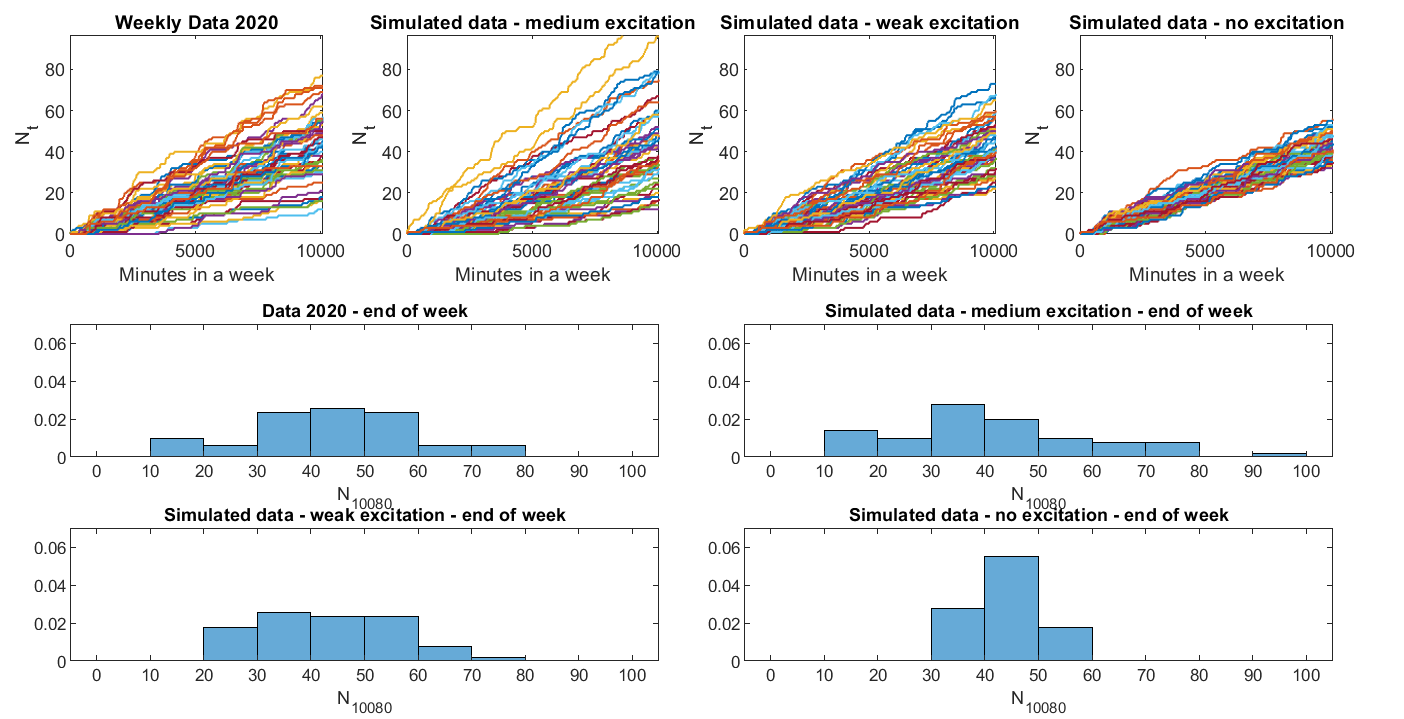}
\caption{Real data study for $\mu_3$, the polynomially decreasing excitation function. Comparison of temporal evolution of $N_{t}$ and the histograms for $N_{10080}$ with different choices of the weights to real data.}
\label{fig:realData31}
\end{figure}

The excitation share can be used to identify weeks with high and low accident risk. If there were few accidents at the beginning of the week, it is likely that there will also be fewer people traveling and causing accidents in the next few days. On the other hand, if there are a lot of accidents at times when there are a lot of accidents, it is likely that there will also be a lot of accidents during the rest of the week. We do not add additional information on the likeliness of accidents into the intensity function, but by construction the stochastic process self-adapts to these situations. Since in the data we do not observe significantly fewer or more accidents at the beginning or the end of the weeks, the accident risk should remain more or less constant. As the excitation share consists of the product of $N_{t}$ and $\mu(t)$ by a linear increase of $N_{t}$, $\mu(t)$ must compensate this increase to obtain a constant accident risk, which is done by a term proportional to $\frac{1}{t}$, see $\mu_3$. In the case of the constant $\mu_1$, accidents at the end of the week are overemphasized, since $N_{t}$ increases monotonically, and so does the corresponding accident risk. For exponentially decaying $\mu_2$, the effect is reversed. The excitation function decreases faster than $N_{t}$ increases, creating the concave shape of $N_t$ with too many accidents at the beginning of the week.
	
Summarising, for the particular chosen data set initially presented in Figure \ref{fig:motivation}, the sinusoidal choice of $\lambda$ and the rational function $\mu_3$ is proved to be appropriate. 

\section{Conclusions}
\label{sect:conclusions}
Motivated by real data for traffic accidents, we have derived a special stochastic process that considers its own history, proving both analytically and numerically that it mimics qualitatively well the typical time evolution of the number of accidents. Its investigation has allowed for a deep study also of additional features, noticeably the statistics of the times of occurrence of the accidents and of the times elapsing between two consecutive accidents. Numerical experiments have revealed further interesting phenomena, that go beyond the theoretical results.  

\bibliographystyle{plain}
\bibliography{biblio}

\newpage
\begin{appendices}

\section{Supplementary Material}
\subsection{The accident process in the case of~\texorpdfstring{$\boldsymbol{\lambda}$}{},~\texorpdfstring{$\boldsymbol{\mu}$}{} constant}
\label{sect:lambda_mu_const.f}
When $\lambda,\,\mu>0$ are constant the intensity $\lambda^\ast$ takes the special form~\eqref{eq:Hawkes.mu_const}. The process $\{N_t,\,t\geq 0\}$ becomes then a particular instance of a genuine Hawkes process. For this reason, it is interesting to address this case in some detail, considering furthermore that now $\lambda,\,\mu\not\in L^1(\R_+)$ and hence several results obtained so far may no longer hold.

First, we notice that $\Lambda(t)=\lambda t$, $M(t)=\mu t$, which provides the following representations of the mean $m$, energy $E$, and variance $\sigma^2$ of the distribution $f$:
\begin{align*}
	m(t) &= \frac{\lambda}{\mu}\left(e^{\mu t}-1\right), \\
	E(t) &= \frac{\lambda}{\mu^2}\left((\lambda+\mu)e^{2\mu t}-(2\lambda+\mu)e^{\mu t}+\lambda\right), \\
	\sigma^2(t) &= E(t)-m^2(t)=\frac{\lambda}{\mu}e^{\mu t}\left(e^{\mu t}-1\right).
\end{align*}
From here, we see that $m,\,E,\,\sigma^2\to +\infty$ when $t\to +\infty$, therefore in this case we expect a spreading of the distribution $f(\cdot,t)$ on $\N$ over time, with accidents becoming more and more numerous.

Thanks to the simplification brought by the constant coefficients $\lambda,\,\mu$, we are even able to provide the exact expression of $f(n,t)$ for every $n\in\N$ and every $t>0$:
\begin{theorem} \label{theo:f.exact}
	Assume $\lambda,\,\mu>0$ constant. Then
	$$ f(n,t)=\frac{\Gamma(\frac{\lambda}{\mu}+n)}{n!\,\Gamma(\frac{\lambda}{\mu})}e^{-\lambda t}\left(1-e^{-\mu t}\right)^n, \qquad n\in\N,\,t>0, $$
	where $\Gamma$ denotes the gamma function.
\end{theorem}
\begin{proof}
	For fixed $t>0$, we proceed by induction on $n$.
	
	For $n=0$ we get $f(0,t)=e^{-\lambda t}$, which coincides with~\eqref{eq:f(0,t)} in the present case of $\lambda$ constant.
	
	We now assume that the formula holds for a certain $n\in\N$ and we check it for $n+1$. Recalling~\eqref{eq:f.recurrence}, we obtain:
	\begin{align*}
		f(n+1,t) &= (\lambda+\mu n)e^{-(\lambda+\mu(n+1))t}\int_0^te^{(\lambda+\mu(n+1))s}f(n,s)\,ds \\
		&= (\lambda+\mu n)\frac{\Gamma(\frac{\lambda}{\mu}+n)}{n!\,\Gamma(\frac{\lambda}{\mu})}e^{-(\lambda+\mu(n+1))t}\int_0^te^{\mu(n+1)s}\left(1-e^{-\mu s}\right)^n\,ds, \\
		\intertext{whence, invoking the binomial theorem to develop $\left(1-e^{-\mu s}\right)^n$,}
		&= (\lambda+\mu n)\frac{\Gamma(\frac{\lambda}{\mu}+n)}{n!\,\Gamma(\frac{\lambda}{\mu})}e^{-(\lambda+\mu(n+1))t}\sum_{j=0}^n\binom{n}{j}(-1)^j\int_0^te^{\mu(n+1-j)s}\,ds \\
		&= (\lambda+\mu n)\frac{\Gamma(\frac{\lambda}{\mu}+n)}{n!\,\Gamma(\frac{\lambda}{\mu})}e^{-\lambda t}\sum_{j=0}^n\binom{n}{j}(-1)^j\frac{e^{-\mu jt}-e^{-\mu(n+1)t}}{\mu(n+1-j)} \\
		\intertext{and further, observing that $\frac{1}{n!}\binom{n}{j}\frac{1}{n+1-j}=\frac{1}{(n+1)!}\binom{n+1}{j}$,}
		&= (\lambda+\mu n)\frac{\Gamma(\frac{\lambda}{\mu}+n)}{(n+1)!\,\Gamma(\frac{\lambda}{\mu})}e^{-\lambda t}\sum_{j=0}^n\binom{n+1}{j}(-1)^j\frac{e^{-\mu jt}-e^{-\mu(n+1)t}}{\mu}. \\
		\intertext{Now we exploit the property of the gamma function that $\Gamma(z+1)=z\Gamma(z)$ for $z\in\mathbb{C}$ with $\Re\text{e}\,z>0$. In particular, letting $z=\frac{\lambda}{\mu}+n$ we have $\Gamma(\frac{\lambda}{\mu}+n)=\frac{\Gamma(\frac{\lambda}{\mu}+n+1)}{\frac{\lambda}{\mu}+n}$, whence}
		&= \frac{\Gamma(\frac{\lambda}{\mu}+n+1)}{(n+1)!\,\Gamma(\frac{\lambda}{\mu})}e^{-\lambda t}\sum_{j=0}^n\binom{n+1}{j}(-1)^j\left(e^{-\mu jt}-e^{-\mu(n+1)t}\right) \\
		&= \frac{\Gamma(\frac{\lambda}{\mu}+n+1)}{(n+1)!\,\Gamma(\frac{\lambda}{\mu})}e^{-\lambda t}\left[\sum_{j=0}^{n+1}\binom{n+1}{j}(-e^{-\mu t})^j\right. \\
		&\phantom{=} \left.-\sum_{j=0}^{n+1}\binom{n+1}{j}(-1)^j\cdot e^{-\mu(n+1)t}\right], \\
		\intertext{where we have extended the sums up to $n+1$ because the terms corresponding to $j=n+1$ cancel with one another. Since $\sum_{j=0}^{n+1}\binom{n+1}{j}(-1)^j=(1-1)^{n+1}=0$, we get finally}
		&= \frac{\Gamma(\frac{\lambda}{\mu}+n+1)}{(n+1)!\,\Gamma(\frac{\lambda}{\mu})}e^{-\lambda t}\left(1-e^{-\mu t}\right)^{n+1}
	\end{align*}
	and the thesis follows by induction.
\end{proof}

From Theorem~\ref{theo:f.exact}, invoking Stirling's formula $\Gamma(x+1)\sim\sqrt{2\pi x}\left(\frac{x}{e}\right)^x$ for $x\in\R$, $x\to +\infty$ and using it with $x=\frac{\lambda}{\mu}+n-1$, we obtain
$$ f(n,t)\sim\frac{n^{\frac{\lambda}{\mu}-1}}{\Gamma(\frac{\lambda}{\mu})}e^{-\lambda t}\left(1-e^{-\mu t}\right)^n \quad (n\to\infty), $$
whence we see that for every $t>0$ the tail of the distribution $f(\cdot,t)$ is slim owing to the exponential decay induced by the last factor in the expression above (notice that $0<1-e^{-\mu t}<1$ for all $t>0$). Thus, at every $t>0$ all statistical moments of $f(\cdot,t)$ are finite. At the same time, $f(n,t)\to 0$ when $t\to +\infty$ for every $n\in\N$, meaning that as time increases the probability mass spreads over the whole $\N$. As a result, the statistical moments of $f(\cdot,t)$ cannot be asymptotically bounded. In particular, the trends of $m$, $E$ examined before imply that the only asymptotically bounded moment is the zeroth order one.

\subsection{The accident times in the case of~\texorpdfstring{$\boldsymbol{\lambda}$}{},~\texorpdfstring{$\boldsymbol{\mu}$}{} constant}
\label{sect:aggregateInterTimes.constant_lambda_mu}
As we did in Section~\ref{sect:lambda_mu_const.f}, it is interesting to specialise the results about the accident times to the case in which the background and self-excitation intensities $\lambda$, $\mu$ are positive constants.

In particular, combining Theorem~\ref{theo:f.exact} and~\eqref{eq:gk} we obtain that the pdf of the accident time $T_k$ reads
$$ g_k(t)=\frac{(\lambda+\mu(k-1))\Gamma(\frac{\lambda}{\mu}+k-1)}{(k-1)!\,\Gamma(\frac{\lambda}{\mu})}e^{-\lambda t}\left(1-e^{-\mu t}\right)^{k-1},
\qquad t>0,\ k\in\N,\,k\geq 1. $$

On the other hand, formula~\eqref{eq:hk} for the pdf of the interaccident time $\Delta{T}_k$ becomes
\begin{align}
	h_k(\tau) &= (\lambda+\mu(k-1))e^{-(\lambda+\mu(k-1))\tau}\int_0^{+\infty}(\lambda+\mu(k-2))f(k-2,s)\,ds, \nonumber \\
	\intertext{i.e., owing to~\eqref{eq:gk} and Proposition~\ref{prop:gk.pdf},}
	&= (\lambda+\mu(k-1))e^{-(\lambda+\mu(k-1))\tau}\int_0^{+\infty}g_{k-1}(s)\,ds \nonumber \\
	&= (\lambda+\mu(k-1))e^{-(\lambda+\mu(k-1))\tau}, \qquad \tau\geq 0,\ k\in\N,\,k\geq 1.
	\label{eq:lambda_mu_const.hk}
\end{align}
We conclude that when $\lambda,\,\mu>0$ are constant the $k$-th interaccident time $\Delta{T}_k$ is \textit{exponentially distributed} with parameter $\lambda+\mu(k-1)$. Nevertheless, also in this case the process $\{N_t,\,t\geq 0\}$ is not a classical renewal (Poisson) process, because the $\Delta{T}_k$'s are still not identically distributed as the parameters of their exponential laws depend on $k$.

Finally, concerning the case of the aggregate interaccident time, from Section~\ref{sect:lambda_mu_const.f} we deduce
$$ m(\cT)=\frac{\lambda}{\mu}\left(e^{\mu\cT}-1\right), $$
therefore~\eqref{eq:h}, together with Theorem~\ref{theo:f.exact} and~\eqref{eq:lambda_mu_const.hk}, implies
\begin{align}
	\label{eq:h_lambdaMuConst}
	\begin{split}
		h(\tau) &= \frac{\mu}{\lambda\left(e^{\mu\cT}-1\right)}\sum_{k=1}^\infty h_k(\tau)(\lambda+\mu(k-1))\int_0^\cT f(k-1,t)\,dt \\
		&= \frac{e^{-\lambda\tau}}{\left(e^{\mu\cT}-1\right)\Gamma(\frac{\lambda}{\mu}+1)}\sum_{k=0}^\infty\frac{(\lambda+\mu k)^2\Gamma(\frac{\lambda}{\mu}+k)}{k!}e^{-\mu k\tau}
		\int_0^\cT e^{-\lambda t}\left(1-e^{-\mu t}\right)^k\,dt.
	\end{split}
\end{align}
As particular cases, which allow for deeper analytical computations, we consider:
\begin{enumerate}[label=\roman*)]
	\item If $\lambda=\mu$, using that $\Gamma(2)=1$, $\Gamma(k)=(k-1)!$, and that the integral can be easily computed by taking $1-e^{-\mu t}$ as new integration variable, the expression above simplifies as
	$$ h(\tau)=\frac{\mu}{e^{\mu\cT}-1}\sum_{k=1}^\infty k\left[e^{-\mu\tau}\left(1-e^{-\mu\cT}\right)\right]^k. $$
	Letting $x:=e^{-\mu\tau}(1-e^{-\mu\cT})$, which is a number comprised in $(0,\,1)$ for every $\cT>0$ and $\tau\geq 0$, the infinite sum can be computed explicitly as
	$$ \sum_{k=1}^\infty kx^k=x\sum_{k=1}^\infty kx^{k-1}=x\cdot\frac{d}{dx}\sum_{k=1}^\infty x^k=\frac{x}{(1-x)^2}, $$
	whence finally
	$$ h(\tau)= \frac{\mu}{e^{\mu\cT}}\cdot\frac{e^{-\mu\tau}}{\left[1+\left(e^{-\mu\cT}-1\right)e^{-\mu\tau}\right]^2}, \qquad \tau\geq 0. $$
	\item If $\lambda>0$ and $\mu=0$ (no self-excitation), from~\eqref{eq:lambda_mu_const.hk} we deduce that every $\Delta{T}_k$ is exponentially distributed with parameter $\lambda$, hence in particular that the $h_k$'s are independent of $k$. As a result, also $\Delta{T}_K$ turns out to be exponentially distributed with parameter $\lambda$:
	\begin{align}
		\label{eq:aggregate_h_exponential}
		h(\tau)=\lambda e^{-\lambda\tau}, \qquad \tau\geq 0.     
	\end{align}
\end{enumerate}

\subsection{Numerical experiments}
\label{sect:numerics}
This section provides a numerical study of the results obtained in Sections~\ref{sect:model.accident_distr} and~\ref{sect:acc_times}. Throughout this section, we use the following discretisations: For the distribution function $f$ describing the probability of $n$ accidents at time $t$ defined in~\eqref{eq:distributionFunction} we consider a time grid $\{t_j\}_{j \in \mathbb{N}}$ with constant stepsize $\Delta t = t_{j+1}-t_j$. The temporal evolution for fixed $n$ is described by the ODE in~\eqref{eq:Boltz_f.strong}, which is discretised using an explicit Euler scheme:
\begin{align*}
    f_{n,j+1} = f_{n,j} + \Delta t \left((\lambda(t_j) + \mu(t_j)(n-1))f_{n-1,j} - (\lambda(t_j) + \mu(t_j)n)f_{n,j}  \right),
\end{align*}
for a sufficiently small stepsize $\Delta t>0$.
The mean number of accidents given in~\eqref{eq:m.closedForm} and of the energy given in~\eqref{eq:E.closedForm} are discretised applying a rectangular rule with stepsize $\Delta t = \tfrac{1}{100}$ to the integral. The same technique is used for the computation of $\Lambda(t)$ and $M(t)$ defined in~\eqref{eq:Lambda_M}. The variances can be derived by computing the difference between the energy and the squared mean.

\subsubsection{Numerical study for different choices of~\texorpdfstring{$\boldsymbol{\lambda}$}{} and~\texorpdfstring{$\boldsymbol{\mu}$}{}}
\label{sec:numStudy}
We provide examples for different choices for $\lambda(t)$ and $\mu(t)$ as defined in \eqref{eq:our_lambda*}. In particular, we study the form of the distribution function in \eqref{eq:distributionFunction} for fixed numbers of accidents $n$ and times $t$. We also investigate the evolution in time of the mean and the variance of accidents.

The accident evolution is significantly different when either $\lambda(t)$ or $\mu(t)$ belongs or not to $L^1(\R_+)$ (cf. Theorem~\ref{theo:moments_bounded}, Corollary~\ref{cor:slimTails}, Proposition~\ref{prop:gk.pdf} and Lemma~\ref{lem:gIsPdf}). For the background intensity $\lambda$ we consider two different choices: $\lambda(t)=0.8(\sin{t}+1)$ (marked in blue), which is non-integrable, and $\lambda(t) = 4e^{-\sqrt{t}}$ (marked in red), which is integrable. 

\begin{figure}[!t]
    \centering
	\includegraphics[width=\linewidth]{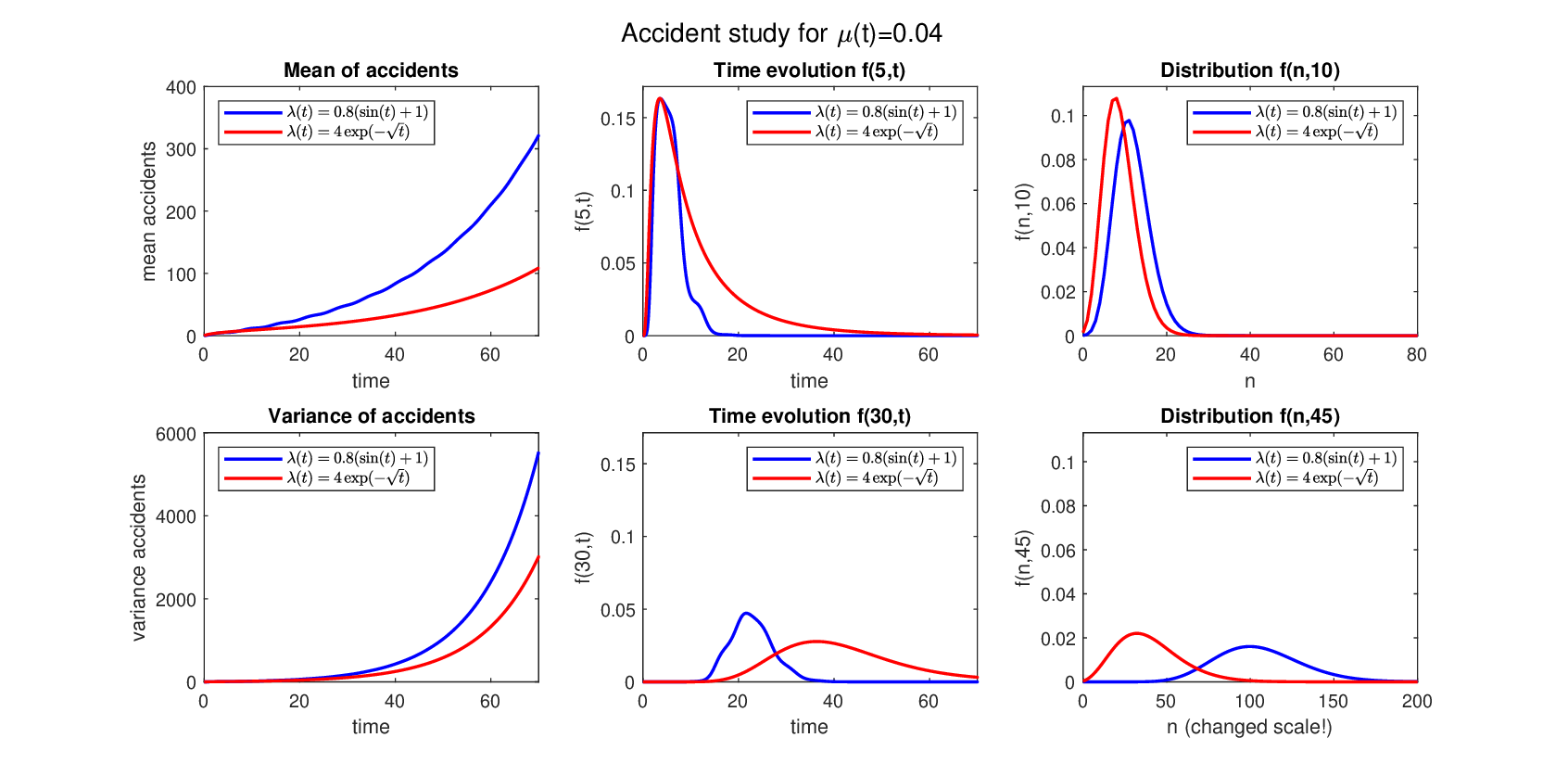}
	\caption{The temporal evolution of the mean and variance of the number of accidents (first column) together with the accident distribution function for $\mu(t) =0.04$ and two different choices of $\lambda(t)$, for fixed $n$ (second column) and fixed $t$ (third column).}
	\label{fig:accidentStudy1}
\end{figure}	

For the excitation intensity, as a representative of a non-integrable function on the positive real axis we consider in Figure~\ref{fig:accidentStudy1} a constant $\mu$. The rationale of this choice is that previous accidents have a constant-in-time impact on the probability of further accidents. This is different from a standard Hawkes process, where typically the past history of the process has a lower and lower impact as time progresses, see e.g.,~\cite{Goettlich.2024}. The panels of the first column of Figure~\ref{fig:accidentStudy1} show the temporal evolution of the mean and the variance of the number of accidents. For both the sinusoidal (blue) and the exponentially decaying (red) background intensity function, the mean and the variance increase in time. In particular, also in the second case, where the background intensity decays to zero, the non-decaying self-excitation intensity drives globally the accident evolution.

The panels of the second column consider instead the distribution function~\eqref{eq:distributionFunction} for a fixed number of accidents $n=5$ (upper panel) and $n=30$ (lower panel). With the sinusoidal background intensity (blue), reaching $n=30$ accidents takes much longer than reaching $n=5$ accidents with also more uncertainty in the latter case, as the lower peak of the distribution function suggests. With the exponentially decaying background intensity (red), $f(30,t)$ is a widely spread function.
	
In the third column we fix the time to $t=10$ (upper panel) and $t=45$ (lower panel) and examine the probability of the number of accidents occurred up to $t$. For the sinusoidal background intensity, more accidents are likely compared to the exponentially decaying background intensity, which is consistent with the observations in the first column.

\begin{figure}[!t]
    \centering
	\includegraphics[width=\linewidth]{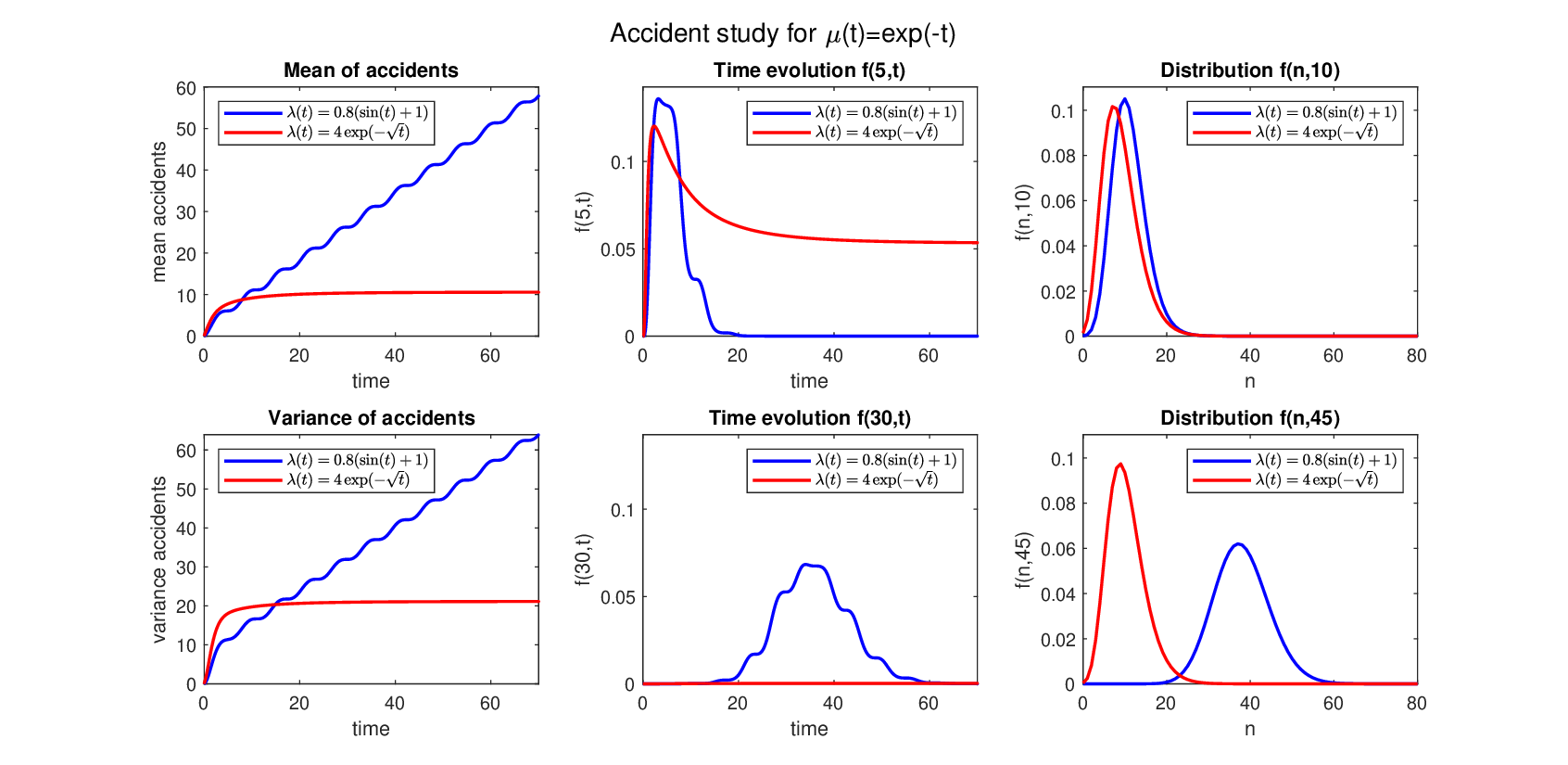}
	\caption{The temporal evolution of the mean and variance of the number of accidents (first column) together with the accident distribution function for $\mu(t) =e^{-{t}}$ and two different choices of $\lambda(t)$, for fixed $n$ (second column) and fixed $t$ (third column).}
	\label{fig:accidentStudy2}
\end{figure}

In Figure~\ref{fig:accidentStudy2} we repeat the same numerical experiments choosing the $L^1$-self-excitation intensity $\mu(t)=e^{-t}$. In the first column, with the sinusoidal background intensity we observe a growth of the mean number of accidents that is linear on average with superimposed sinusoidal fluctuations and no longer convex as in the first example. With the exponentially decaying background intensity, the mean increases and then approaches some maximum value close to $10$. Similar observations can be made for the variance, thereby underpinning the theoretical result on the uniform boundedness of the moments of $f$ for $\lambda,\,\mu\in L^1(\R_+)$, cf. Theorem~\ref{theo:moments_bounded}.

In the second column, with the exponentially decaying background intensity the distribution function $f(5,\cdot)$ of $n=5$ accidents is non-vanishing as $t\to +\infty$. As the probability of further accidents vanishes as $t\to +\infty$, a strictly positive probability of having exactly $n=5$ accidents remains. Conversely, for $n=30$ accidents the curve stays very close to zero for all times, as $30$ accidents are very unlikely to happen. This is also emphasised in the third column, where one can observe that the red curve moves only slightly rightwards. In particular, Corollary~\ref{cor:slimTails} predicts slim tails of $f(n,10)$ and $f(n,45)$, which will be further elaborated in Section~\ref{sec:boundsExample}.
	
\begin{figure}[!t]
    \centering
	\includegraphics[width=\linewidth]{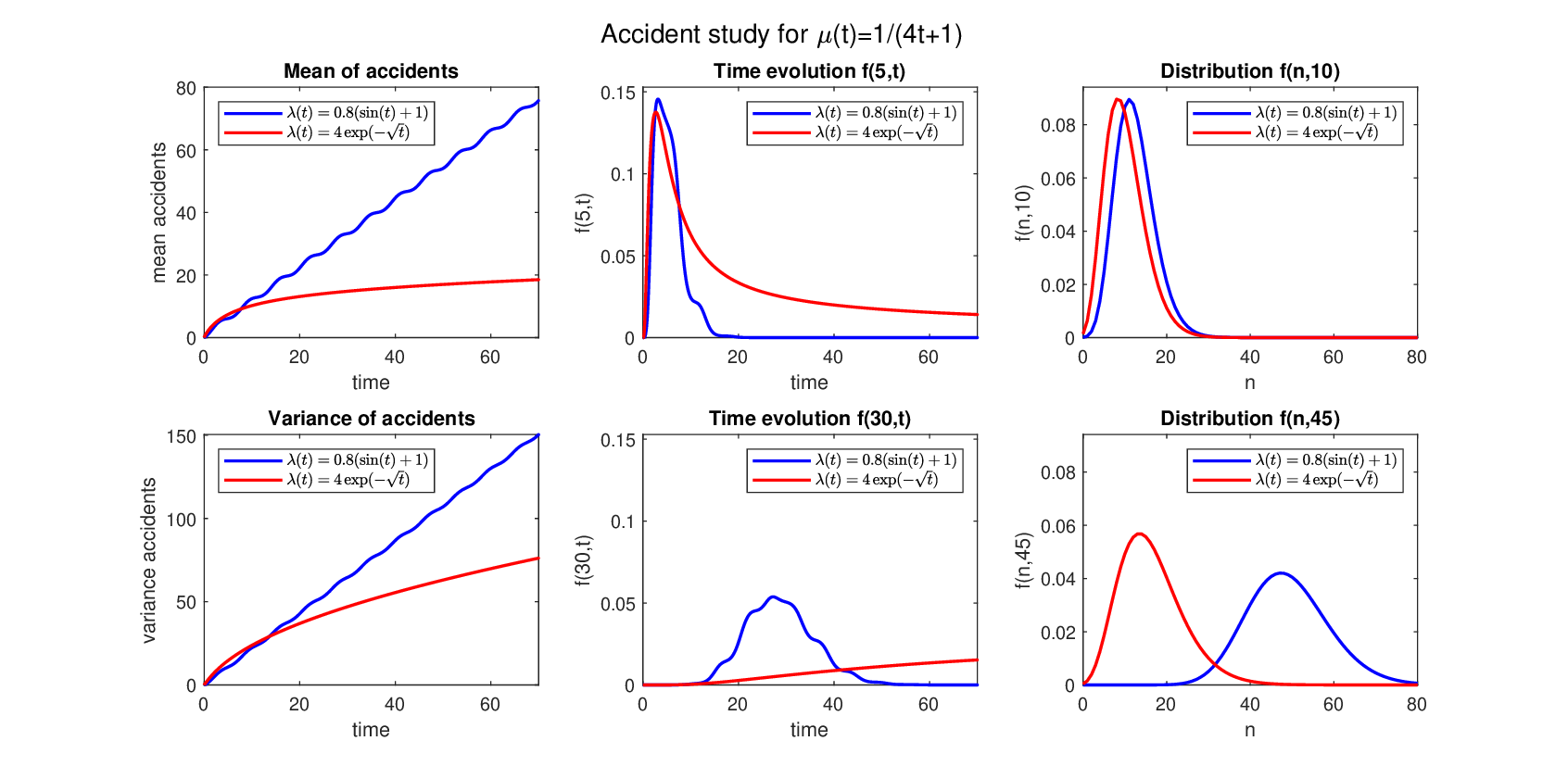}
	\caption{The temporal evolution of the mean and variance of the number of accidents (first column) together with the accident distribution function for $\mu(t) = \frac{1}{4t+1}$ and two different choices of $\lambda(t)$, for fixed $n$ (second column) and fixed $t$ (third column).}
	\label{fig:accidentStudy3}
\end{figure}

As a last example, in Figure~\ref{fig:accidentStudy3} we present the borderline case of a self-excitation intensity $\mu\notin L^1(\R_+)$ that is however close to being integrable, in particular $\mu(t)=\frac{1}{4t+1}$. Although this $\mu$ is non-integrable like the one considered in Figure~\ref{fig:accidentStudy1}, the resulting accident process resembles rather that of~Figure \ref{fig:accidentStudy2}. However, unlike Figure \ref{fig:accidentStudy2}, the mean number of accidents grows unboundedly in time and furthermore $f(30,t)$ with the exponentially decaying background intensity (red) does not stay close to zero for all times.
	
\subsubsection{Validation of the bounds on the distribution function}
\label{sec:boundsExample}
\begin{figure}[!t]
    \centering
	\includegraphics[width=0.97\linewidth]{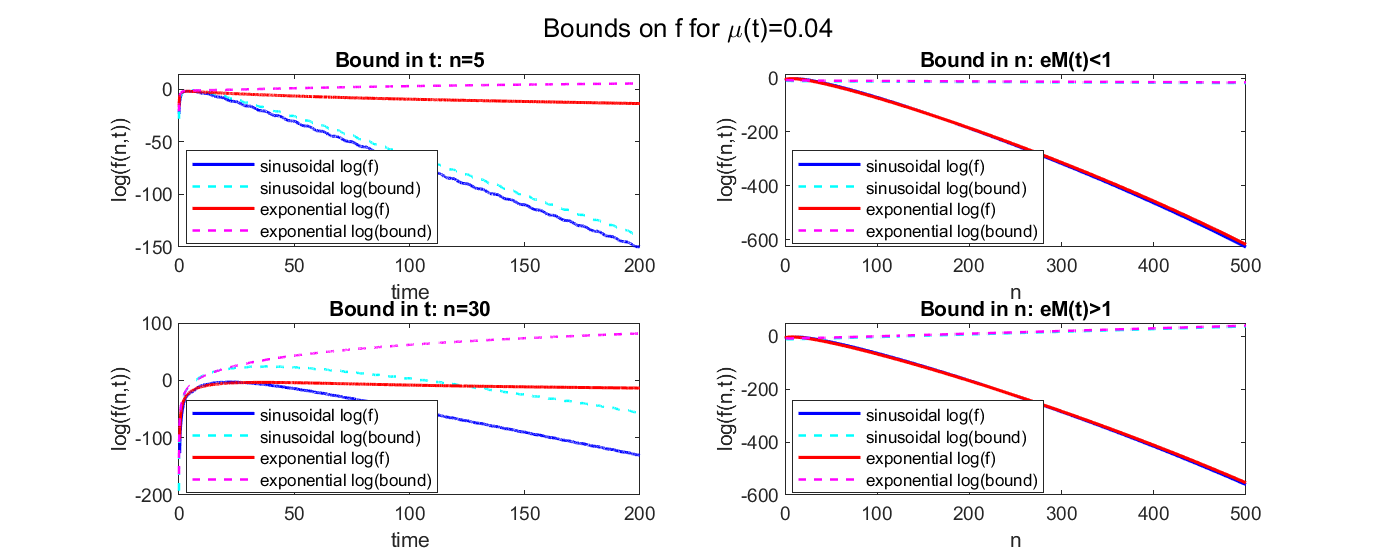}
	\caption{Comparison of the distribution function in time (left) and the number of accidents (right) for $\mu(t) = 0.04$ together with bounds \eqref{eq:f.bound} (left) and \eqref{eq:bound.f.Stirling} (right).}
	\label{fig:bounds1}
\end{figure}
\begin{figure}[!t]
    \centering
	\includegraphics[width=0.97\linewidth]{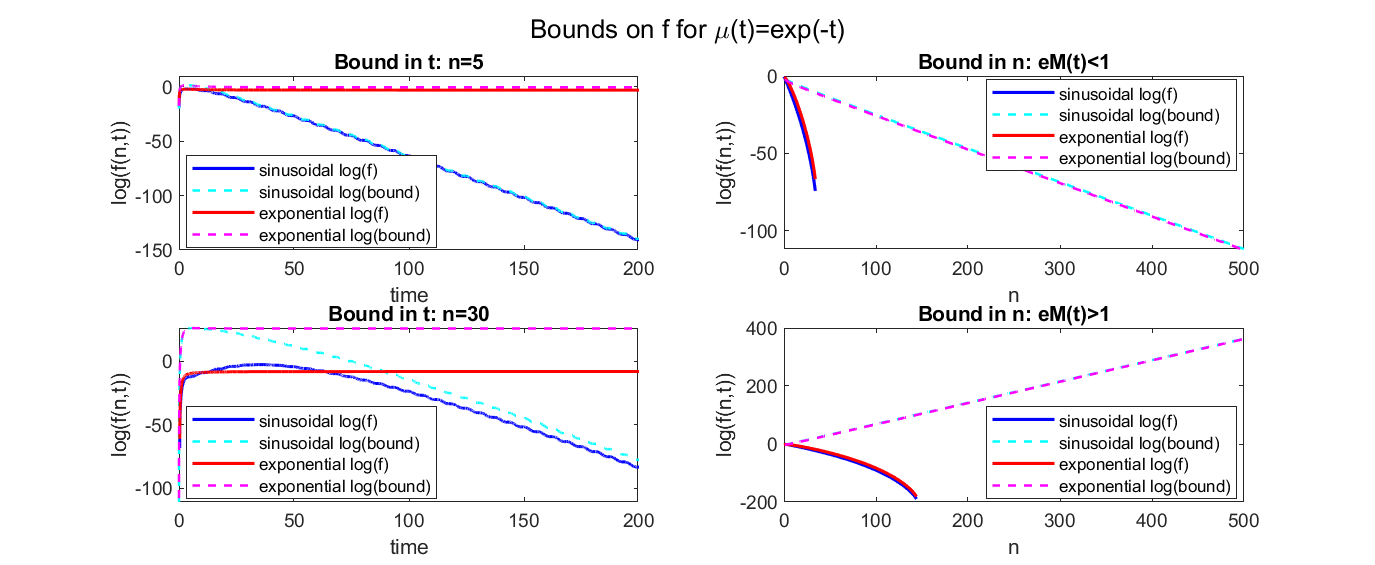}
	\caption{Comparison of the distribution function in time (left) and the number of accidents (right) for $\mu(t)=e^{-t}$ together with bounds \eqref{eq:f.bound} (left) and \eqref{eq:bound.f.Stirling} (right).}
	\label{fig:bounds2}
\end{figure}

This section aims to show the quality of the estimates obtained in Section~\ref{sec:f.bounds}. For this, we refer to the two first examples of the previous section with the different self-excitation intensities $\mu$, each combined with the two background intensities $\lambda$ seen before. We anticipate that, for fixed $n$, the bound established in Theorem~\ref{theo:f.bound} and~\eqref{eq:f.bound} turns out to be numerically a good estimate in time of the tail of the distribution function. The analysis is carried out in the left panels of Figures~\ref{fig:bounds1},~\ref{fig:bounds2}. The bounds are plotted in dashed lines in cyan (for the sinusoidal $\lambda$) and magenta (for the exponential $\lambda$). The vertical axis shows the logarithmic values of $f$ and its bounds, in order to better illustrate the tail behaviour. In the right panels of the same figures we plot instead the bound in $n$ for two fixed time instants, also on the logarithmic scale. Instead of using the bound from Theorem~\ref{theo:f.bound} and~\eqref{eq:f.bound}, here we use the approximate description of the bound developed in~\eqref{eq:bound.f.Stirling}, which is more intuitive than~\eqref{eq:f.bound}.

Figure~\ref{fig:bounds1} refers to the case of the constant self-excitation intensity $\mu(t)=0.04$. As also presented in Figure \ref{fig:accidentStudy1}, for fixed $n$ we have diminishing values of $f(n,\cdot)$ when $t\to +\infty$, which is accompanied by the diminishing bound for the sinusoidal case (blue). Considering the bound in~\eqref{eq:bound.f.Stirling} for fixed time $t$ on the right of Figure \ref{fig:bounds1}, we choose $t=9.09$ to obtain $eM(t)<1$ for both background intensities. Although the bounds in cyan and magenta decrease, they describe only loosely the tail behaviour for $n\to\infty$, as their decay is much too slow. For $t=10.19$ we have $eM(t)>1$ and increasing bounds, however the probability of having many accidents in a finite time still decreases as $n\to\infty$.

For the exponentially decaying self-excitation intensity $\mu(t)=e^{-t}$ we provide the analogous illustrations of the bounds in Figure~\ref{fig:bounds2}. Choosing $\lambda(t)=4e^{-\sqrt{t}}$ the observations for the bound in $t$ are similar to the example considered in Figure~\ref{fig:bounds1}. Conversely, sharper bounds are observed with this $\lambda$ (magenta) for fixed numbers of accidents $n=5$ and $n=30$. Considering the bound in $n$, Corollary~\ref{cor:slimTails} states slim tails of the distribution function $f(\cdot,t)$ for fixed $t$, due to the integrability of $\lambda$ and $\mu$. This is also observed in Figure~\ref{fig:bounds2}. For $t=0.35$ condition $eM(t)<1$ holds in both examples and the bounds are decreasing but at a significantly lower rate than the distribution functions $f$ itself. Conversely, for $t=1.45$ we are in the case $eM(t)>1$ and the bound is even increasing.

The bound behaviour in the case of the self-excitation intensity $\mu(t)=\frac{1}{4t+1}$ are structurally very similar to the presented examples.

\subsubsection{Aggregate interaccident time}
\label{sect:num.Intermediate}
Following the discussion of Section~\ref{sect:aggregateInterTimes}, we provide a comparison of the aggregate interaccident time obtained by a simulation of $\{N_t,\,t\geq 0\}$ using a thinning procedure~\cite{Ogata1981} and by the pdf $h$ defined in~\eqref{eq:def.h},~\eqref{eq:h}. As derived in~\eqref{eq:aggregate_h_exponential}, with a constant background intensity $\lambda$ and without self-excitation the aggregate interaccident time is exponentially distributed. For this reason, we always compare the results to the pdf of the best-fit exponential distribution. This allows us to evaluate how close the obtained distributions are to the most basic case. The distributions depend on the time horizon, which is chosen as $\cT=80$ (see e.g.,~\eqref{eq:h}).

\begin{figure}[!t]
    \centering
    \includegraphics[width=0.9\linewidth]{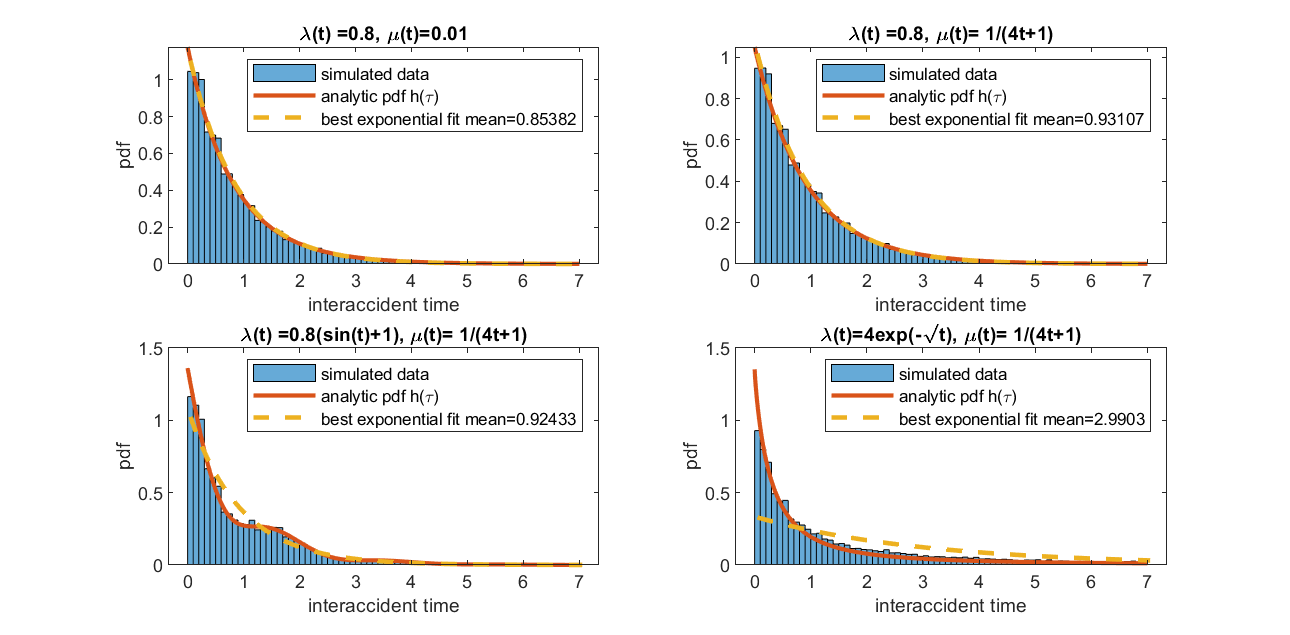}
    \caption{Histogram and pdf for the aggregate interaccident times for different choices of $\lambda$ and $\mu$ compared to the pdf of a best-fit exponential distribution.}
    \label{fig:aggregateIntertimes}
\end{figure}

Figure~\ref{fig:aggregateIntertimes} presents four examples of the pdf of the aggregate interaccident time for different choices of $\lambda$ and $\mu$. The illustration to the top left shows constant $\lambda(t)=0.08$ and $\mu(t)=0.01$, for which in Section~\ref{sect:aggregateInterTimes.constant_lambda_mu} and in~\eqref{eq:h_lambdaMuConst} we obtained an explicit form of $h$. The exponential distribution that fits the accident process best is very close to $h$, which allows us to conclude that for constant $\lambda$ and $\mu$ as well as a sufficiently large time horizon the distribution of the aggregate interaccident times is close to an exponential distribution. 

A similar conclusion can be drawn from the illustration in the top right of Figure~\ref{fig:aggregateIntertimes}, where $\lambda$ is unchanged while $\mu(t)=\frac{1}{4t+1}$. The distribution is different from the first one, but can still be approximated quite well by an exponential distribution. This is no longer true for the examples in the bottom row of Figure~\ref{fig:aggregateIntertimes}, when considering a time-varying background intensity. The histogram and the pdf $h$ differ significantly from the exponential best-fit. In the last example (bottom right), we choose $\lambda(t)=4\exp(-\sqrt{t})\in L^1(\mathbb{R}_+)$, which also leads to a pdf significantly different from an exponential distribution. The reason here is that the probability of further accidents vanishes as $t\to +\infty$ (see Figure~\ref{fig:accidentStudy3}) leading to a very long interaccident time. Therefore, the pdf in this example has a stronger tail than in the previous examples and therefore it does not fit an exponential distribution.
\end{appendices}
\end{document}